\documentclass[11pt]{article}
\usepackage[a4paper,margin=1.25in]{geometry} 
\usepackage{nccmath}

\usepackage[title,toc,titletoc]{appendix}
\usepackage{etoolbox}
\usepackage{pdflscape}
\patchcmd{\appendices}{\quad}{. }{}{}
\usepackage{graphicx}
\usepackage{subfiles}
\usepackage{booktabs}
\usepackage{float}
\usepackage{xcolor,soul}
\usepackage{amsmath,amssymb,amsfonts}
\usepackage{amsthm}
\usepackage{tabu}
\usepackage{enumitem}
\usepackage[english]{babel}
\usepackage{mathrsfs}
\usepackage{mathtools}
\usepackage[mathscr]{euscript}
\usepackage{subfig}
\usepackage{adjustbox}
\usepackage{multirow}
\usepackage{xr}
\usepackage{tcolorbox}
\usepackage{rotating}
\usepackage{float}
\usepackage{setspace}
\selectfont

\usepackage{footmisc}
\usepackage{natbib}
\usepackage[resetlabels]{multibib}
\newcites{web}{Useful websites}

\usepackage{thmtools}
\usepackage[colorlinks,citecolor=blue,urlcolor=blue]{hyperref}
\usepackage[capitalize]{cleveref}

\numberwithin{equation}{section}
\renewcommand{\arraystretch}{1.2}

\def\bs{\ensuremath\boldsymbol}

\DeclareMathOperator*{\argmin}{arg\,min}

\newcommand\smallO{
	\mathchoice
	{{\scriptstyle\mathcal{O}}}%
	{{\scriptstyle\mathcal{O}}}
	{{\scriptscriptstyle\mathcal{O}}}%
	{\scalebox{.7}{$\scriptscriptstyle\mathcal{O}$}}%
}

\makeatletter
\def\blfootnote{\xdef\@thefnmark{}\@footnotetext}
\makeatother

\theoremstyle{plain}
\newtheorem{theorem}{Theorem}[section]

\theoremstyle{definition}
\newtheorem{assumption}{Assumption}
\newtheorem{remark}{Remark}

\theoremstyle{definition}

\newtheorem*{AssumptionER*}{Assumption ER*}
\newtheorem*{AssumptionER**}{Assumption ER**}
\newtheorem*{AssumptionER*(i)}{Assumption ER*(i)}
\newtheorem*{AssumptionER*(ii)}{Assumption ER*(ii)}

\newtheorem*{AssumptionCM*}{Assumption CM*}

\newenvironment{AssumptionAD}{%
	\begingroup
	\begin{assumption}%
	}{%
	\end{assumption}%
	\endgroup
}

\newenvironment{AssumptionADD}{%
	\begingroup
	\begin{assumption}%
	}{%
	\end{assumption}%
	\endgroup
}

\newenvironment{AssumptionJKD}{%
	\begingroup
	\begin{assumption}%
	}{%
	\end{assumption}%
	\endgroup
}

\newenvironment{AssumptionAD*}{%
	\begingroup
	\begin{assumption}%
	}{%
	\end{assumption}%
	\endgroup
}

\newenvironment{AssumptionJK*}{%
	\begingroup
	\begin{assumption}%
	}{%
	\end{assumption}%
	\endgroup
}

\newenvironment{AssumptionAD**}{%
	\begingroup
	\begin{assumption}%
	}{%
	\end{assumption}%
	\endgroup
}

\newenvironment{AssumptionJK}{%
	\begingroup
	\begin{assumption}%
	}{%
	\end{assumption}%
	\endgroup
}

\newenvironment{AssumptionJK**}{%
	\begingroup
	\begin{assumption}%
	}{%
	\end{assumption}%
	\endgroup
}

\declaretheoremstyle[bodyfont=\slshape]{slshape}

\newlist{thmlist}{enumerate}{1}
\setlist[thmlist]{label=\textup{(\roman{thmlisti})},
	ref=\thethm\textup{(\roman{thmlisti})}}

\Crefname{theorem}{Theorem}{Theorems}
\Crefname{lemma}{Lemma}{Lemmas}
\Crefname{proposition}{Proposition}{Propositions}
\Crefname{corollary}{Corollary}{Corollaries}
\Crefname{assumption}{Assumption}{Assumptions}
\Crefname{remark}{Remark}{Remarks}
\Crefname{thm}{Theorem}{Theorems}
\Crefname{lem}{Lemma}{Lemmas}

\newtheoremstyle{exnum}
{0pt}{0pt}
{\itshape}
{}
{\bfseries}
{.}
{ }
{\thmname{#1}\thmnumber{ {#2}}\thmnote{ (#3)}}

\theoremstyle{exnumred}
\newtheorem{example}{Example}

\newenvironment{exampleagain}[2][]{%
	\begin{trivlist}
		\setlength{\topsep}{0pt}%
		\setlength{\partopsep}{0pt}%
		\setlength{\itemsep}{0pt}%
		\setlength{\parsep}{0pt}%
		\item[\hskip\labelsep\bfseries Example~{\ref{#2}}%
		\if\relax\detokenize{#1}\relax\else\ (#1)\fi.]%
		\itshape
	}{%
	\end{trivlist}
}

\newlength{\exblocksep}   
\newlength{\exrulesep}    
\newlength{\exrulethick}  

\newif\ifexamplehrprev
\examplehrprevfalse
\newif\ifinexamplehr
\inexamplehrfalse

\makeatletter
\@ifundefined{AddToHook}{%
	\newtoks\examplehreverypar
	\examplehreverypar=\everypar
	\everypar=\expandafter{\the\examplehreverypar
		\ifinexamplehr\else\global\examplehrprevfalse\fi
	}%
}{%
	\AddToHook{para/begin}{%
		\ifinexamplehr\else\global\examplehrprevfalse\fi
	}%
}
\makeatother

\pretocmd{\section}{\global\examplehrprevfalse}{}{}
\pretocmd{\subsection}{\global\examplehrprevfalse}{}{}
\pretocmd{\subsubsection}{\global\examplehrprevfalse}{}{}

\newcommand{\examplehrstart}{%
	\global\inexamplehrtrue
	
	\par\addvspace{\ifexamplehrprev 0pt\else \exblocksep\fi}%
	\begingroup
	\setlength{\parskip}{0pt}%

	\ifexamplehrprev\else
	\noindent\rule{\linewidth}{\exrulethick}\par
	\fi
	\vspace{\exrulesep}%
}

\newcommand{\examplehrend}{%
	\par\vspace{\exrulesep}%
	\noindent\rule{\linewidth}{\exrulethick}\par
	\endgroup
	\par\addvspace{\exblocksep}%
	\global\inexamplehrfalse
	\global\examplehrprevtrue
	\ignorespacesafterend
}

\newenvironment{examplehr}[1][]{%
	\examplehrstart
	\begin{example}[#1]\normalfont
	}{%
	\end{example}%
	\examplehrend
}

\newenvironment{exampleagainhr}[2][]{%
	\examplehrstart
	\begin{exampleagain}[#1]{#2}%
	}{%
	\end{exampleagain}%
	\examplehrend
}

\usepackage{tikz}
\usetikzlibrary{calc,decorations.pathreplacing}

\begin{document}
	
\title{Jackknife Inference for Fixed Effects Models}
\vspace{-0.2cm}
\author{Ayden Higgins$^{\dagger}$\\University of Exeter}
\date{\today$^\ddagger$}
\maketitle
\vspace{-1cm}

\begin{abstract}
This paper develops a general method of inference for fixed effects models which is (i) automatic, (ii) computationally inexpensive, (iii) tuning parameter-free, and (iv) highly model agnostic. Specifically, we show how to combine a collection of subsample estimators into a jackknife $t$-statistic, from which hypothesis tests, confidence intervals, and $p$-values are readily obtained. \\

\noindent \textbf{Keywords:} Fixed effects, Jackknife bias correction, Jackknife $t$-statistics.\\
\textbf{JEL classification:} C12, C18, C23.\\	
\end{abstract}

\blfootnote{\hspace{-0.2in}$^\dagger$Address: University of Exeter Business School, Rennes Drive, Exeter, Devon, EX4 4PU, United Kingdom. Email: \texttt{a.higgins@exeter.ac.uk}.}

\blfootnote{\hspace{-0.2in}$^\ddagger$First Version: February 25, 2026. }

\newpage
\section{Introduction}
Panel models with fixed effects are ubiquitous in empirical research. Such models allow researchers to account for unobserved heterogeneity by introducing unit-specific parameters. For example, individual fixed effects may absorb time-invariant differences across individuals, such as ability, risk preferences, or innate productivity, while time fixed effects may capture aggregate shocks and common trends, such as business cycle fluctuations and policy changes that affect all units. Over time, the fixed effects models researchers wish to estimate have become increasingly complex, including more elaborate fixed effects specifications, such as time-varying group effects and interactive effects, models which feature lagged outcomes and other predetermined regressors, nonlinear and nonstandard models such as binary outcome models and quantile regression, and extensions to panels with three or more dimensions. While it is frequently possible to adapt existing estimation methods to encompass such models, statistical inference has become increasingly challenging, as such estimators typically exhibit highly complex statistical properties. 

This notwithstanding, it is often reasonable to expect that such estimators, once suitably centred and scaled, are asymptotically normal with a number of leading order bias terms. Indeed, the accumulated work of several decades has given us a good sense of the likely structure of such results, even when explicit formulae are unavailable. For example, in the context of a generic panel model with one-way fixed effects, we typically expect a standard estimator $\hat{\varphi}$ of a (scalar) common parameter $\varphi$ to be asymptotically distributed 
\begin{align}
	\sqrt{NT}( \hat{\varphi} - \varphi )
	\xrightarrow{d}
	\mathcal{N}( \mu, \sigma^2 ), \label{oneway}
\end{align}
as both dimensions of the panel $N,T \rightarrow \infty$ and $N/T \rightarrow \gamma^2 \in (0, \infty)$. Similarly, for a generic panel model with two-way fixed effects, we typically expect a standard estimator to be asymptotically distributed 
\begin{align}
	\sqrt{NT}( \hat{\varphi} - \varphi )
	\xrightarrow{d}
	\mathcal{N}( \mu_1 + \mu_2, \sigma^2 ), \label{twoway}
\end{align}
under the same asymptotic regime, where we distinguish $\mu_{1}$ and $\mu_{2}$ by their origin in either cross-sectional or time fixed effects. This paper develops a method of inference designed for such situations where, rather than working with model and estimator specific expressions, our method rests upon two high-level assumptions. First, we assume the existence of an estimator of a parameter of interest which, once centred and scaled, is asymptotically normally distributed with a number of \textit{unknown} bias terms and an \textit{unknown} variance. Second, we assume the existence of a collection of subsample estimators whose joint distribution (along with the full sample estimator) can be characterised through a bias matrix and a covariance matrix, both of which are implied by the known subsampling design. Under these conditions we derive a jackknife $t$-statistic which provides a general method of inference for fixed effects models which is automatic, computationally inexpensive, tuning parameter-free, and highly model agnostic.

To illustrate the simplicity of our statistics, at a most basic level, we expect for a generic panel model with one-way (cross-sectional) fixed effects that
\begin{align}
	\mathcal{J}
	\coloneqq
	\frac{ 4 \hat{\varphi}^{(0)} - \hat{\varphi}^{(1)} - \hat{\varphi}^{(2)} - 2 \varphi }{ | \hat{\varphi}^{(1)} - \hat{\varphi}^{(2)} | }
	\xrightarrow{d}
	t_1, \label{eq1}
\end{align}
as $N,T \rightarrow \infty$ and $N/T \rightarrow \gamma^2 \in (0, \infty)$, where $\mathcal{J}$ denotes our jackknife $t$-statistic, $\hat{\varphi}^{(0)}$ denotes an estimator obtained using the full sample, $\hat{\varphi}^{(1)}$ and $\hat{\varphi}^{(2)}$ denote subsample estimators obtained by partitioning the data into the first $1, \dots, T/2$ time periods, and the last $T/2+1, \ldots, T$ time periods, respectively, and $t_1$ denotes the $t$-distribution with one degree of freedom.\footnote{We presume that $T$ is even. The case of time fixed effects is symmetric and we instead partition the cross-section.} Corresponding to \eqref{eq1}, a $( 1 - \alpha )$ two-sided confidence interval for $\varphi$ is obtained as
\begin{align}
	\left[
	 2 \hat{\varphi}^{(0)}
	- \frac{1}{2}\hat{\varphi}^{(1)}
	- \frac{1}{2}\hat{\varphi}^{(2)}
	\pm
	t_{ 1, 1 - \alpha / 2 }
	\frac{1}{2}
	|
	\hat{\varphi}^{(1)} - \hat{\varphi}^{(2)}
	|
	\right], \notag
\end{align}
where $t_{ 1, 1 - \alpha / 2 }$ denotes the $( 1 - \alpha / 2 )$ quantile of the $t$-distribution with one degree of freedom. In parallel, for a generic panel model with two-way (cross-sectional and time) fixed effects we expect that
\begin{align}
	\mathcal{J}
	\coloneqq
	\frac{ 6 \hat{\varphi}^{(0)}
		- \hat{\varphi}^{(1)}
		- \hat{\varphi}^{(2)}
		- \hat{\varphi}^{(3)}
		- \hat{\varphi}^{(4)}
		- 2 \varphi
	}{
		\left|
		\hat{\varphi}^{(2)} - \hat{\varphi}^{(1)}
		\right|
	}
	\xrightarrow{d}
	t_1, \label{eq2}
\end{align}
as $N, T \rightarrow \infty$ and $N / T \rightarrow \gamma^2 \in ( 0, \infty )$, where $\hat{\varphi}^{(0)},\hat{\varphi}^{(1)}$, and $\hat{\varphi}^{(2)}$ are as described previously, and $\hat{\varphi}^{(3)}$, $\hat{\varphi}^{(4)}$ denote subsample estimators obtained by partitioning the data into the first $1, \ldots, N / 2$ cross-sectional units, and the last $N / 2 + 1, \ldots, N$ cross-sectional units, respectively.\footnote{We presume that both $N$ and $T$ are even.} Corresponding to \eqref{eq2}, a $( 1 - \alpha )$ two-sided confidence interval for $\varphi$ is obtained as
\begin{align}
	\left[
	3 \hat{\varphi}^{(0)}
	- \frac{1}{2} \hat{\varphi}^{(1)}
	- \frac{1}{2} \hat{\varphi}^{(2)}
	- \frac{1}{2} \hat{\varphi}^{(3)}
	- \frac{1}{2} \hat{\varphi}^{(4)}
	\pm
	t_{1,1-\alpha/2}
	\frac{1}{2}
	\left|
	\hat{\varphi}^{(2)} - \hat{\varphi}^{(1)}
	\right|
	\right].
	\notag 
\end{align}

\subsection{Related literature}
The incidental parameter problem (IPP), first articulated by \citet{neyman}, arises in models where the number of parameters grows with sample size. For panel models with fixed effects this growth is inherent, since the number of fixed effects increases with $N$ and/or $T$.\footnote{For ease of exposition we focus initially on two-dimensional panels. See Section \ref{KWM} for discussion of multi-dimensional panels.} The central consequence of the IPP is that estimators of low-dimensional objects of interest, such as slope coefficients and average partial effects, can exhibit nonnegligible bias. Under large $N$, fixed $T$ asymptotics, this bias can be so severe that it results in inconsistency. Under large $N$, large $T$ asymptotics, standard estimators are typically consistent, but their limiting distributions are not centred at zero, referred to as asymptotic bias. As a result, conventional methods of inference built upon known, zero-centred limiting distributions are no longer valid; in particular standard Wald, likelihood ratio, and Lagrange multiplier tests, as well as confidence intervals and $p$-values, will generally be incorrect unless the bias is explicitly accounted for; see, e.g., \citet{fernandez2018fixed}. This has motivated a large literature which seeks to understand the consequences of the IPP in specific settings and to develop valid methods of inference.

The most standard route to valid inference is to eliminate asymptotic bias so that conventional large sample approximations apply and, in particular, to do so by analytic means. Classically, this is achieved by characterising (i.e., deriving a formula for) the bias, and using this to construct a ``bias-corrected'' estimator.\footnote{More broadly, related analytic strategies modify the objective function or score; see, e.g., \cite{arellano2016likelihood}, \cite{jochmans2019likelihood}, and \cite{bonhomme2024neyman}.} The main limitation of such approaches is, however, that the appropriate modifications needed to bias-correct, construct a debiased estimator, or otherwise eliminate asymptotic bias are model and estimator-specific, meaning that these need to be derived on a case by case basis. For sufficiently standard problems it is possible to appeal to quite broad results which have been obtained for one-way \citep{hahn2004jackknife,hahn2011bias}, two-way  \citep{fernandez2016individual}, and three-way fixed effects models \citep{czarnowske2025debiased}. However, the resulting expressions are very cumbersome, and require a researcher to adapt these to their specific setting, usually by computing layers of derivatives and/or application of the delta method.\footnote{For example, when the parameter of interest is a function of both the fixed effects and common parameters.} Moreover, their implementation is often dependent on tuning parameters to which their performance can be highly sensitive. Beyond these concerns, as fixed effects models have become more complex, results have become increasingly challenging to derive, and, for that matter, increasingly tedious, leading to an absence of results in empirically relevant cases. As a consequence, recent years have seen the widespread adoption of automatic methods of inference for fixed effects models.

The jackknife provides one such approach. \citet{hahn2004jackknife} proposed leave-one-out jackknife bias correction for static panel models with one-way fixed effects. This has been extended to models with two-way fixed effects in \cite{hughes2026jackknife}. While such constructions may exhibit favourable higher-order properties (see, e.g. \cite{hahn2024efficient}), such schemes are not appropriate when data exhibit dependence, which is common in panel data, and, moreover, command a computational cost comparable with the bootstrap due to their large number of auxiliary estimations. As an alternative, \citet{dhaene2015split} propose a philosophically distinct approach called the split-panel jackknife. Their method is particularly attractive because it is both computationally inexpensive (requiring as little as two auxiliary estimations) and highly model agnostic; in particular, it accommodates dependence. The split-panel jackknife has been extended to two-way fixed effects models in \cite{fernandez2016individual} and has become the default method of bias correction in many applications. However, jackknife bias correction, by whichever means, does not deliver a complete method of inference on its own: a bias-corrected estimator must be paired with a separate (typically analytic) variance estimator in order to conduct inference. In this sense, existing jackknife approaches provide only a \textit{semi}-automatic method of inference.

The bootstrap can provide an alternative \textit{fully}-automatic method of inference for fixed effects models. In particular, if it is possible to devise a bootstrap scheme which correctly replicates the asymptotic distribution of an estimator, including any bias, then inference can be based on the bootstrap distribution without the need to make further adjustments; see, e.g, \citet{gonccalves2015bootstrap} and \citet{higgins2024bootstrap}. However, the task of designing such a scheme is nontrivial. For example, \citet{gonccalves2015bootstrap} show that common bootstrap schemes match the asymptotic variance of the least squares estimator, while failing to replicate its asymptotic bias, in the context of a dynamic panel model with individual fixed effects. Although in this case valid bootstrap schemes can be constructed (see further results in \citet{gonccalves2015bootstrap} and \citet{higgins2025inference}), more generally it must be verified that any proposed scheme does indeed yield valid inference for a given application. Moreover, the bootstrap incurs a computational cost which, although sometimes overstated, cannot be dismissed in the context of fixed effects models due to the large sample sizes and the corresponding large number of parameters.

In this paper we return to the jackknife and develop a complete method of inference underpinned by the ``small number of large subsamples'' philosophy of the split-panel jackknife. To achieve our purpose, we first construct a minimum-variance unbiased jackknife (MVUJ) estimator of a parameter of interest, obtained as an unbiased linear combination of the full sample and subsample estimators that minimises asymptotic variance. We then use this to construct a jackknife $t$-statistic, which provides a straightforward means to conduct inference, even when the statistical properties of an estimator are highly complex. Beyond our simplest constructions, we develop jackknife $t_q$-statistics with $q > 1$ which yield sharper inference, at the expense of additional complexity, and ultimately derive a jackknife counterpart to Hotelling's $t^2$-statistic \citep{hotelling1931generalization} which allows us to extend our method of inference to encompass vector parameters and to accommodate joint hypothesis tests. Overall we provide a general framework for the systematic development of jackknife bias-corrected estimators and of jackknife $t$-statistics.  

\section{Jackknife Inference}
Suppose we observe a panel data set with entries indexed by $i = 1, \ldots, N$ and $t = 1, \ldots, T$, and wish to perform inference on a scalar parameter $\varphi$ in the presence of a large number of fixed effects, which are present in one or more dimensions of the panel. The parameter of interest $\varphi$ might be, for example, an individual slope coefficient, a linear combination of the form $\varphi \coloneqq \bs{a}^\top \bs{\beta}$, where $\bs{\beta}$ is a vector of primitive model parameters, or a more complicated function of primitive parameters (including fixed effects), such as an average partial effect. We make the following assumption. 

\begin{AssumptionAD}\label{AAD}
There exists an estimator of $\varphi$, denoted by $\hat{\varphi}^{(0)}$, such that
\begin{align}
		r_{NT}( \hat{\varphi}^{(0)} - \varphi )
		=
		z_{NT}
		+
		\sum_{r=1}^R \mu_{r,NT}
		+
		\smallO_{p}( 1 ), \label{reds}
\end{align}
as $N,T \rightarrow \infty$ and $N/T \rightarrow \gamma^2 \in ( 0, \infty )$, where $\mu_{1,NT}, \ldots, \mu_{R,NT}$ are nonstochastic, but possibly $N$ and $T$ dependent, and 
	\begin{align}
		z_{NT}
		\xrightarrow{d}
		\mathcal{N}( 0, \sigma^2 ). \notag 
	\end{align} 
\end{AssumptionAD}
Assumption \ref{AAD} assumes the existence of an estimator $\hat{\varphi}^{(0)}$ that is asymptotically normal, and which may feature one or more bias terms.\footnote{Note that we do not assume $\mu_{r,NT}$ to be bounded nor convergent.} Expansions of the form \eqref{reds} have been established for a large range of models, estimators, and parameters of interest. Examples include least squares (LS) and generalized method of moments (GMM) estimators of linear dynamic panel models with individual fixed effects \citep[][]{nickell,hahn2002asymptotically,arellano}, extremum estimators of nonlinear models with individual fixed effects such as logit and probit models \citep[][]{hahn2004jackknife,fernandez2009fixed,hahn2011bias}, nonlinear models with individual and time fixed effects \citep{fernandez2016individual}, and LS and maximum likelihood (ML) estimators of linear and nonlinear models with interactive fixed effects \citep[][]{bai2009panel,moon_dynamic_nodate,chen2021nonlinear}, amongst many others. When ${\varphi}$ is a common parameter, or a sufficiently smooth function thereof, one typically has $r_{NT} = \sqrt{NT}$. When $\varphi$ is a function of both common parameters and fixed effects, it is possible that the rate may be slower; see \citet{fernandez2018fixed} for further discussion.

\begin{remark}
	Many standard estimators are asymptotically linear, meaning that the first term in the representation \eqref{reds} is a normalised sum of the form
	\begin{align}
		z_{NT}
		=
		\frac{1}{a_{NT}} \sum_{(i,t) \in S} \nu_{it},
		\notag
	\end{align}
	where $a_{NT}$ denotes an appropriate normalising sequence and $S$ denotes a set of indices. For example, when the object of interest is a common parameter one typically has
	\begin{align}
		z_{NT}
		=
		\frac{1}{\sqrt{NT}} \sum_{(i,t) \in S_0} \nu_{it},
		\notag
	\end{align}
	with $S_0 \coloneqq \{ 1, \ldots, N \} \times \{ 1, \ldots,  T \}$. Asymptotic normality of $z_{NT}$ then follows through a central limit theorem. 
\end{remark}

\begin{remark}
	Though many results of the form \eqref{reds} have been established, there are many more cases of interest where we might expect that such results could, in principle, be obtained. For example, a distributional result may have been established with only strictly exogenous regressors, however, one might reasonably suspect that a similar result also holds when lagged outcomes appear as regressors, with existing theory suggesting the likely structure (in terms of bias) even if it does not deliver the exact expressions.
\end{remark}

We now introduce some examples to illustrate Assumption \ref{AAD} that we shall return to throughout the remainder of the paper.

\begin{examplehr}[One-way Effects]\label{ex:onefe}\small
	Consider a dynamic regression model where outcomes are generated according to
	\begin{align}
		y_{i t}
		=
		\lambda_{i} + \varphi y_{i,t - 1} + \varepsilon_{it},
		\qquad
		i = 1, \ldots, N,
		\quad
		t =	1, \ldots, T,
		\notag
	\end{align}
	where $\varepsilon_{it} \sim \textup{iid}\ \mathcal{N}( 0, 1 )$, $\varepsilon_{it}$ is independent of $\lambda_{i}$, and $| \varphi | < 1$. Let $\hat{\varphi}^{(0)}$ be the LS estimator of $\varphi$. As
	$N, T \rightarrow \infty$ and $N / T \rightarrow \gamma^2 \in ( 0, \infty )$,
	\begin{align}
		\sqrt{ N T }
		\big( \hat{\varphi}^{(0)} - \varphi \big)
		&=
		\frac{ ( 1 - \varphi^2 ) }{ \sqrt{NT} }
		\sum_{ (i,t) \in S_0}
		 y_{i,t - 1} \varepsilon_{it}
		-
		\sqrt{\frac{N}{T}} ( 1 + \varphi )
		+
		\smallO_p( 1 ), \label{firstq}
	\end{align}
	where $S_0 \coloneqq \{ 1, \ldots, N \} \times \{ 1, \ldots, T \}$, and
	\begin{align}
		\frac{( 1 - \varphi^2 )}{\sqrt{NT}}
		\sum_{ (i,t) \in S_0}
		y_{i,t - 1} \varepsilon_{it}
		\xrightarrow{d}
		\mathcal{N} \big( 0, 1 - \varphi^2 \big).
		\notag
	\end{align}
	Hence, Assumption \ref{AAD} holds with
	\begin{align}
		z_{N T}
		=
		\frac{ (1 - \varphi^2) }{ \sqrt{NT} }
		\sum_{ (i,t) \in S_0 }
		 y_{i,t - 1} \varepsilon_{it},
		\notag
	\end{align}
	$r_{NT} = \sqrt{NT}$, $R = 1$, $\mu_{NT} = - \sqrt{\frac{N}{T}}( 1 + \varphi )$, and $\sigma^2 = 1 - \varphi^2$.
	\normalsize
\end{examplehr}

\begin{examplehr}[Two-way Effects]\label{ex:twofe}\small
	Consider a panel probit model where outcomes are generated according to
	\begin{align}
		y_{it}
		&=
		\mathbf{1}\left\{ \lambda_i + \gamma_t + \varphi d_{it} \geq \varepsilon_{it} \right\},
		\qquad
		i = 1, \ldots, N,
		\quad
		t = 1, \ldots, T,
		\notag
	\end{align}
	where $\varepsilon_{it} \sim \textup{iid}\ \mathcal{N}(0,1)$, $\lambda_i \sim \text{iid}$, $\gamma_t  \sim \text{iid}$, and $d_{it} \sim \text{iid}$ are mutually independent. Let $\hat{\varphi}^{(0)}$ be the ML estimator of $\varphi$. Define
	\begin{align}
		\varpi_{it} \coloneqq \lambda_i + \gamma_t + \varphi d_{it}, \quad
		H_{it}      \coloneqq \frac{ \phi( \varpi_{it} ) }{ \Phi( \varpi_{it} ) ( 1 - \Phi( \varpi_{it} ) )}, \quad \text{and} \quad
		\omega_{it} \coloneqq \phi(\varpi_{it}) H_{it},	
		\notag 
	\end{align} 
	where $\Phi( \cdot )$ and $\phi( \cdot )$ denote the CDF and PDF of the standard normal distribution, respectively. 
	Let $\tilde{d}_{it}$ denote the residual from the projection of $d_{it}$ on the space spanned by the individual and time effects weighted by $\omega_{it}$, and define
	\begin{align}
		\sigma^2
		\coloneqq
		\lim_{N,T \rightarrow \infty}
		\mathbb{E} 
		\left[
		\frac{1}{NT} \sum_{i=1}^{N} \sum_{t=1}^T
		\omega_{it} \tilde d_{it}^{2}
		\right]^{-1}. \footnotemark  \notag
	\end{align}	
	Using Theorem 4.1 in \citet{fernandez2016individual} (see also Theorem C.1 for the underlying expansion), as
	$N, T \rightarrow \infty$ with $N / T \rightarrow \gamma^2 \in ( 0, \infty )$, \footnotetext{That is, let
		$( \bs{a}^\ast, \bs{g}^\ast ) \coloneqq \argmin_{ \bs{a}, \bs{g} } \frac{1}{NT} \sum^N_{i=1}\sum^T_{t=1} \mathbb{E}[ \omega_{it}( d_{it} - a_i - g_t )^2 ]$ (under a suitable normalisation for $a_i + g_t$). Then $\tilde{d}_{it} \coloneqq d_{it} - a_i^{\ast} - g_t^{\ast}$; see \citet{fernandez2016individual}.}
	\begin{align}
		\sqrt{NT} \big( \hat{\varphi}^{(0)} - \varphi \big)
		=
		\frac{ \sigma^2 }{ \sqrt{NT} } \sum_{ (i,t) \in S_0 } \tilde{d}_{it} H_{it} ( y_{it} - \Phi( \varpi_{it} ) )
		+ 
		\sqrt{ \frac{N}{T} } \sigma^2 b_1 + \sqrt{ \frac{T}{N} } \sigma^2 b_2
		+ 
		\smallO_p( 1 ), \label{thisth}
	\end{align}
	where $S_0 \coloneqq \{ 1, \ldots, N \} \times \{ 1, \ldots, T \}$,
	\begin{align}
		\frac{ \sigma^2 }{ \sqrt{NT} } \sum_{ (i,t) \in S_0 } \tilde{d}_{it} H_{it} ( y_{it} - \Phi( \varpi_{it} ) )
		\xrightarrow{d}
		\mathcal{N}( 0, \sigma^2 ), \notag
	\end{align}
	\begin{align}
		b_1
		\coloneqq
		\mathbb{E} \left[
		\frac{1}{2N} \sum_{i=1}^{N}
		\frac{ \sum_{t=1}^T \omega_{it} \tilde{d}_{it}^{2} }{ \sum_{t=1}^T \omega_{it} }
		\right] \varphi,
		\qquad \text{and} \qquad
		b_2
		\coloneqq
		\mathbb{E} \left[
		\frac{1}{2T}\sum_{t=1}^T
		\frac{ \sum_{i=1}^{N} \omega_{it} \tilde{d}_{it}^{2} }{ \sum_{i=1}^{N} \omega_{it} }
		\right]\varphi. \notag
	\end{align}
	Hence, Assumption \ref{AAD} holds with
	\begin{align}
		z_{NT}
		=
		\frac{ \sigma^2 }{ \sqrt{NT} } \sum_{ (i,t) \in S_0} \tilde{d}_{it} H_{it} ( y_{it} - \Phi( \varpi_{it} ) ), \notag
	\end{align}
	$r_{NT} = \sqrt{NT}$, $R = 2$, $\mu_{1,NT} = \sqrt{\frac{N}{T}} \sigma^2 b_1$, $\mu_{2,NT} = \sqrt{\frac{T}{N}} \sigma^2 b_2$, and $\sigma^2$ as defined above.
	\normalsize
\end{examplehr}

We now introduce our second assumption which concerns the existence and joint asymptotic distribution of subsample estimators.
\begin{AssumptionJK}\label{AJK}
	Let $m \geq R + 2$ be a fixed integer. There exist $m - 1$ estimators $\hat{\varphi}^{(1)}, \ldots, \hat{\varphi}^{( m - 1 )}$ obtained from subsamples of the data such that 
	\begin{align}\label{eq:AJKk}
		r_{NT}
		\left(
		\bs{\hat{\varphi}}
		-
		\varphi \bs{\iota}_m
		\right)
		&=
		\bs{z}_{NT}
		+
		\bs{A} \bs{\mu}_{NT}
		+
		\smallO_p( 1 ),
	\end{align}
	as $N, T \rightarrow \infty$ and $N / T \rightarrow \gamma^2 \in ( 0, \infty )$, where
	\begin{align}
		\bs{z}_{NT} \xrightarrow{d} \mathcal{N} \left( \bs{0}_m, \sigma^2 \bs{C} \right), \notag
	\end{align}
	 $\bs{\hat{\varphi}} \coloneqq (\hat{\varphi}^{(0)}, \ldots, \hat{\varphi}^{( m - 1 )})^\top$, $\bs{A}$ and $\bs{C}$ are known matrices of dimension $m \times R$ and $m \times m$, respectively, $\textup{rank}(\bs{A}) = R$, $\bs{\iota}_m \notin \textup{col}( \bs{A} )$, $\textup{null}( (\bs{A}, \bs{\iota}_m )^\top) \not\subseteq \textup{null}(\bs{C})$, $\bs{C} \succeq \bs{0}_{m \times m}$, and $\bs{\mu}_{NT} \coloneqq (\mu_{1,NT}, \ldots, \mu_{R,NT})^\top$.
\end{AssumptionJK}

Assumption \ref{AJK} assumes the existence of subsample estimators $\hat{\varphi}^{(1)}, \ldots, \hat{\varphi}^{(m-1)}$ that exhibit different linear combinations of the same bias terms $\mu_{1,NT}, \ldots, \mu_{R,NT}$. These different linear combinations are represented in the matrix $\bs{ A }$. The matrix $\bs{ C }$ encodes the dependence structure amongst the subsamples. Both $\bs{A}$ and $\bs{C}$ are treated as known because they are determined by a known, deterministic subsampling scheme, but otherwise Assumption \ref{AJK} is deliberately agnostic about underlying model and estimation approach. 

To provide some motivation behind Assumption \ref{AJK}, consider a generic two-way fixed effects model. Suppose an estimator admits an expansion of the form
\begin{align}
	\sqrt{NT}
	\big(
	\hat{\varphi}^{(0)}
	-
	\varphi
	\big)
	&=
	\frac{1}{\sqrt{NT}}
	\sum_{ (i,t) \in S_0 }
	\nu_{ it }
	-
	\sqrt{\frac{N}{T}} b_1
	-
	\sqrt{\frac{T}{N}} b_2
	+
	\smallO_p( 1 ),
	\notag
\end{align}
where $b_1$ and $b_2$ are nonstochastic, and $\nu_{it}$ is mean zero. Consider a subsample which takes $T_1$ time-series observations and $N_1$ cross-sectional observations, and let $\hat{\varphi}^{(1)}$ denote the corresponding estimator computed on this subsample. One typically obtains a corresponding expansion
\begin{align}
	\sqrt{N_1 T_1}\big( \hat{\varphi}^{(1)} - \varphi \big)
	&=
	\frac{1}{\sqrt{N_1 T_1}} \sum_{(i,t) \in S_1} \nu_{it}
	-
	\sqrt{ \frac{N_1}{T_1} } b_1
	-
	\sqrt{ \frac{T_1}{N_1} } b_2
	+
	\smallO_{p}( 1 ),
	\notag
\end{align}
as $N_1, T_1 \rightarrow \infty$ and $N_1 / T_1 \rightarrow \gamma_1^2 \in ( 0, \infty )$, where $S_1 \coloneqq \{ 1, \ldots, N_1 \} \times \{ 1, \ldots, T_1 \}$.\footnote{For data which exhibit time-series dependence, consecutive observations may be needed in order to preserve the bias.} Assume, for simplicity, that $\nu_{it}$ are identically and independently distributed with
$\mathbb{E}[ \nu_{it} ] = 0$, $\mathbb{V}[ \nu_{it} ] = \sigma^2 \in ( 0, \infty )$, and $\mathbb{E}\big[ | \nu_{it} |^{ 2 + \delta } \big] < \infty$ for some $\delta > 0$. Moreover, assume $N_1 / N \rightarrow \gamma_N \in ( 0, 1 ]$ and $T_1 / T \rightarrow \gamma_T \in ( 0, 1 ]$ as $N, T \rightarrow \infty$.
Under these conditions,
\begin{align}
	\frac{1}{\sqrt{NT}}
	\begin{pmatrix}
		1 & 0 \\
		0 & \frac{NT}{N_1 T_1}
	\end{pmatrix}
	\begin{pmatrix}
		\sum_{(i,t) \in S_0} \nu_{it} \\[2pt]
		\sum_{(i,t) \in S_1} \nu_{it}
	\end{pmatrix}
	&\xrightarrow{d}
	\mathcal{N}\left(
	\begin{pmatrix}
		0 \\
		0
	\end{pmatrix},
	\sigma^2
	\begin{pmatrix}
		1 & 1 \\
		1 & ( \gamma_N \gamma_T )^{-1}
	\end{pmatrix}
	\right). \notag
\end{align}
Define $r_{NT} \coloneqq \sqrt{NT}$,
\begin{align}
	\bs{\hat{\varphi}}
	\coloneqq
	\begin{pmatrix}
		\hat{\varphi}^{(0)} \\
		\hat{\varphi}^{(1)}
	\end{pmatrix},
	\qquad
	\text{and} 
	\qquad
	\bs{\mu}_{NT}
	\coloneqq
	-
	\begin{pmatrix}
		\sqrt{ \frac{N}{T} } b_1 \\[2pt]
		\sqrt{ \frac{T}{N} } b_2
	\end{pmatrix}.
	\notag
\end{align}
We then obtain
\begin{align}
	r_{NT}\big( \bs{\hat{\varphi}} - \varphi \bs{\iota}_2 \big)
	&=
	\bs{z}_{NT}
	+
	\bs{A}\bs{\mu}_{NT}
	+
	\smallO_{p}( 1 ),
	\notag
\end{align}
where
\begin{align}
	\bs{z}_{NT}
	\coloneqq
	\frac{1}{\sqrt{NT}}
	\begin{pmatrix}
		1 & 0 \\
		0 & \frac{NT}{N_1 T_1}
	\end{pmatrix}
	\begin{pmatrix}
		\sum_{(i,t) \in S_0} \nu_{it} \\[2pt]
		\sum_{(i,t) \in S_1} \nu_{it}
	\end{pmatrix},
	\qquad
	\bs{A}
	\coloneqq
	\begin{pmatrix}
		1 & 1 \\[2pt]
		\frac{T}{T_1} & \frac{N}{N_1}
	\end{pmatrix},
	\notag
\end{align}
and
\begin{align}
	\bs{z}_{NT}
	\xrightarrow{d}
	\mathcal{N}\big( \bs{0}_2, \sigma^2 \bs{C} \big),
	\qquad \text{with} \qquad
	\bs{C}
	=
	\begin{pmatrix}
		1 & 1 \\
		1 & ( \gamma_N \gamma_T )^{-1}
	\end{pmatrix}.
	\notag
\end{align}

\begin{remark}
	In the illustration above there are two distinct bias terms and therefore $m \geq 2$ is required for $\text{rank}( \bs{A} ) = R$. However, when $m = 2$ the additional requirement $\bs{\iota}_m \notin \text{col}( \bs{A} )$ cannot hold: if $\text{rank}( \bs{A}) = 2$ and $m = 2$, then $\text{col}( \bs{A} ) = \mathbb{R}^2$, and necessarily $\bs{\iota}_2 \in \text{col}( \bs{A} )$. Therefore a further subsample estimator is required to satisfy Assumption \ref{AJK} in this example.
\end{remark}

\begin{remark}
		In the illustration above, the assumption that $\nu_{it}$ are identically and independently distributed is adopted for convenience; Assumption \ref{AJK} neither requires that the data be identically nor independently distributed. 
\end{remark}

\begin{remark}
	By taking $m$ to be fixed in Assumption \ref{AJK} we rule out delete-one / leave-one-out jackknife schemes, and more generally designs with $m \to \infty$. We focus here on such schemes due to their computational simplicity and ability to accommodate dependence, leaving a unified treatment for future work.
\end{remark}

We now return to our examples to provide further illustration of Assumption \ref{AJK}.

\begin{exampleagainhr}[One-way Effects]{ex:onefe} \small
	From \eqref{firstq} we have
	\begin{align}
		r_{NT}( \hat{\varphi}^{(0)} - \varphi )
		=
		\frac{1}{\sqrt{NT}}
		\sum_{ (i,t)\in S_0 } \nu_{it}
		+
		\mu_{NT}
		+
		\smallO_p(1),
		\notag
	\end{align}
	as $N, T \rightarrow \infty$ and $N / T \rightarrow \gamma^2 \in ( 0, \infty )$, where $r_{NT} = \sqrt{NT}$, $\nu_{it}\coloneqq ( 1 - \varphi^2 ) y_{i,t-1} \varepsilon_{it}$, and $\mu_{N T} = - \sqrt{\frac{N}{T}} ( 1 + \varphi )$. Assume that $T = 2 T_2$ where $T_2$ is an integer. 
	Define subsamples
	\begin{align}
		S_1 &\coloneqq \{ 1, \ldots, N \} \times \{ 1, \ldots, T_2 \}, \notag \\
		S_2 &\coloneqq \{ 1, \ldots, N \} \times \{ T_2 + 1, \ldots, T \}. \notag
	\end{align}
	Let $\hat{\varphi}^{(j)}$ be the LS estimator using subsample $S_j$ for $j = 1, 2$. One then obtains
	\begin{align}
		\sqrt{NT_2}\big( \hat{\varphi}^{(1)} - \varphi \big)
		=&\
		\frac{1}{\sqrt{N T_2}}
		\sum_{ (i,t) \in S_1 } \nu_{it}
		+
		\sqrt{2} \mu_{NT}
		+
		\smallO_p(1), \notag
	\end{align}
	and
	\begin{align}
		\sqrt{NT_2}\big( \hat{\varphi}^{(2)} - \varphi \big)
		=&\
		\frac{1}{\sqrt{N T_2}}
		\sum_{ (i,t) \in S_2 } \nu_{it}
		+
		\sqrt{2} \mu_{NT}
		+
		\smallO_p(1), \notag
	\end{align}
	as $N, T \rightarrow \infty$ and $N / T \rightarrow \gamma^2$. Let $\bs{\hat{\varphi}} \coloneqq ( \hat{\varphi}^{(0)}, \hat{\varphi}^{(1)}, \hat{\varphi}^{(2)} )^\top$.
	Using the above,
	\begin{align}
		r_{NT}
		\big( \bs{\hat{\varphi}} - \varphi \bs{\iota}_3 \big)
		=
		\frac{1}{\sqrt{NT}}
		\begin{pmatrix}
			1 & 0 & 0  \\
			0 & 2 & 0  \\
			0 & 0 & 2
		\end{pmatrix}
		\begin{pmatrix}
			\sum_{(i,t)\in S_0} \nu_{it} \\
			\sum_{(i,t)\in S_1} \nu_{it} \\
			\sum_{(i,t)\in S_2} \nu_{it}
		\end{pmatrix}
		+
		\mu_{NT}
		\begin{pmatrix}
			1 \\
			2 \\
			2
		\end{pmatrix}
		+
		\smallO_p( 1 ).
		\notag
	\end{align}
	Equivalently,
	\begin{align}
		r_{NT}
		\big( \bs{\hat{\varphi}} - \varphi \bs{\iota}_3 \big)
		=
		\bs{z}_{NT}
		+
		\bs{A} \bs{\mu}_{NT}
		+
		\smallO_p( 1 ),
		\notag
	\end{align}
	where
	\begin{align}
		\bs{z}_{NT}
		\coloneqq
		\frac{1}{\sqrt{NT}}
		\begin{pmatrix}
			1 & 0 & 0  \\
			0 & 2 & 0  \\
			0 & 0 & 2
		\end{pmatrix}
		\begin{pmatrix}
			\sum_{(i,t)\in S_0} \nu_{it} \\
			\sum_{(i,t)\in S_1} \nu_{it} \\
			\sum_{(i,t)\in S_2} \nu_{it}
		\end{pmatrix},
		\qquad
		\bs{z}_{NT} \xrightarrow{d} \mathcal{N}( \bs{0}_3, \sigma^2 \bs{C} ),
		\notag
	\end{align}
	$\bs{\mu}_{NT} \coloneqq \mu_{NT}$,
	\begin{align}
		\bs{C}
		\coloneqq
		\begin{pmatrix}
			1 & 1 & 1\\
			1 & 2 & 0\\
			1 & 0 & 2
		\end{pmatrix},
		\qquad \text{and} \qquad
		\bs{A}
		\coloneqq
		\begin{pmatrix}
			1 \\
			2 \\
			2
		\end{pmatrix}.\footnotemark 
		\notag
	\end{align}
Finally, we verify the remaining conditions of Assumption \ref{AJK}. First, $\textup{rank}( \bs{A} ) = 1$, and $\bs{\iota}_3 \notin \textup{col}( \bs{A} )$ because $\bs{\iota}_3$ is not a scalar multiple of $\bs{A}$. Next, it is easily verified that $\bs{C}$ is positive semidefinite with eigenvalues $0$, $2$, and $3$. Let $\bs{D} \coloneqq ( \bs{A}, \bs{\iota}_3 )$. Solving $\bs{u}^\top \bs{A} = 0$ and $\bs{u}^\top \bs{\iota}_3 = 0$ yields $\textup{null}( \bs{D}^\top ) = \textup{span}\big( ( 0, 1, -1 )^\top \big)$. For $\bs{u}_1 \coloneqq ( 0, 1, -1 )^\top$ we have $\bs{u}_1^\top \bs{C} \bs{u}_1 = 4 > 0$, whereby $\textup{null}( \bs{D}^\top ) \not\subseteq \textup{null}( \bs{C} )$. \footnotetext{With a slight abuse of notation as regards $\bs{\mu}_{NT}$.}
\end{exampleagainhr}

\begin{exampleagainhr}[Two-way Effects]{ex:twofe}\small
	From \eqref{thisth}, 
	\begin{align}
		r_{NT} \big( \hat{\varphi}^{(0)} - \varphi \big)
		=
		\frac{1}{\sqrt{ N T }}
		\sum_{(i, t) \in S_0} \nu_{it}
		+ \mu_{1,NT}
		+ \mu_{2,NT}
		+ \smallO_p( 1 ), \notag
	\end{align}
	as $N,T \rightarrow \infty$ and $N/T \rightarrow \gamma^2 \in ( 0, \infty )$, where $r_{NT} = \sqrt{NT}$, $\nu_{it} \coloneqq \sigma^2 \tilde{d}_{it} H_{it} ( y_{it} - \Phi( \varpi_{it} ) )$, and $\mu_{1,NT}$, $\mu_{2,NT}$ are defined in Example~\ref{ex:twofe}. Next, assume that $N = 2 N_2$ and $T = 2T_2$ where $N_2$ and $T_2$ are integers.
	Define subsamples
	\begin{align}
		S_1 &\coloneqq \{ 1, \ldots, N \} \times \{ 1, \ldots, T_2 \}, \notag \\
		S_2 &\coloneqq \{ 1, \ldots,N\} \times \{ T_2 + 1, \ldots, T\}, \notag \\
		S_3 &\coloneqq \{ 1, \ldots,N_2\} \times \{ 1, \ldots, T\}, \notag \\
		S_4 &\coloneqq \{ N_2 + 1, \ldots, N \} \times \{ 1, \ldots, T \}. \notag
	\end{align}
	Let $\hat{\varphi}^{(j)}$ be the ML estimator using subsample $S_{j}$ for $j = 1, 2, 3, 4$. One then obtains
	\begin{align}
		\sqrt{NT_2} \big( \hat{\varphi}^{(1)} - \varphi \big)
		=&\
		\frac{1}{\sqrt{NT_2}} \sum_{(i,t) \in S_1} \nu_{it}
		+
		\sqrt{2} \mu_{1,NT}
		+
		\frac{1}{\sqrt{2}} \mu_{2,NT}
		+
		\smallO_{p}( 1 ), \notag
	\end{align}
	\begin{align}
		\sqrt{NT_2} \big( \hat{\varphi}^{(2)} - \varphi \big)
		=&\
		\frac{1}{\sqrt{NT_2}} \sum_{(i,t) \in S_2} \nu_{it}
		+
		\sqrt{2} \mu_{1,NT}
		+
		\frac{1}{\sqrt{2}} \mu_{2,NT}
		+
		\smallO_{p}( 1 ), \notag
	\end{align}
	\begin{align}
		\sqrt{N_2T} \big( \hat{\varphi}^{(3)} - \varphi \big)
		=&\
		\frac{1}{\sqrt{N_2T}} \sum_{(i,t) \in S_3} \nu_{it}
		+
		\frac{1}{\sqrt{2}} \mu_{1,NT}
		+
		\sqrt{2} \mu_{2,NT}
		+
		\smallO_{p}( 1 ), \notag
	\end{align}
	and
	\begin{align}
		\sqrt{N_2T} \big( \hat{\varphi}^{(4)} - \varphi \big)
		=&\
		\frac{1}{\sqrt{N_2T}} \sum_{(i,t) \in S_4} \nu_{it}
		+
		\frac{1}{\sqrt{2}} \mu_{1,NT}
		+
		\sqrt{2} \mu_{2,NT}
		+
		\smallO_{p}( 1 ), \notag
	\end{align}
 	as $N, T \rightarrow \infty$ and $N / T \rightarrow \gamma^2$. Now define $\bs{\hat{\varphi}} \coloneqq ( \hat{\varphi}^{(0)}, \hat{\varphi}^{(1)}, \hat{\varphi}^{(2)}, \hat{\varphi}^{(3)}, \hat{\varphi}^{(4)} )^\top$. Using the above
	\begin{align}
		r_{NT} \big( \bs{\hat{\varphi}} - \varphi \bs{\iota}_5 \big)
		=
		\bs{z}_{NT}
		+
		\bs{A} \bs{\mu}_{NT}
		+
		\smallO_{p}( 1 ), \notag
	\end{align}
	where
	\begin{align}
		\bs{z}_{NT}
		\coloneqq
		\frac{1}{\sqrt{NT}}
		\begin{pmatrix}
			1 & 0 & 0 & 0 & 0 \\
			0 & 2 & 0 & 0 & 0 \\
			0 & 0 & 2 & 0 & 0 \\
			0 & 0 & 0 & 2 & 0 \\
			0 & 0 & 0 & 0 & 2
		\end{pmatrix}
		\begin{pmatrix}
			\sum_{(i,t) \in S_0} \nu_{it} \\
			\sum_{(i,t) \in S_1} \nu_{it} \\
			\sum_{(i,t) \in S_2} \nu_{it} \\
			\sum_{(i,t) \in S_3} \nu_{it} \\
			\sum_{(i,t) \in S_4} \nu_{it}
		\end{pmatrix},
		\qquad
		\bs{z}_{NT}
		\xrightarrow{d}
		\mathcal{N}\big( \bs{0}_5, \sigma^2 \bs{C} \big),
		\notag
	\end{align}
	$\bs{\mu}_{NT} \coloneqq ( \mu_{1,NT}, \mu_{2,NT} )^\top$,
	\begin{align}
		\bs{A}
		\coloneqq
		\begin{pmatrix}
			1 & 1 \\
			2 & 1 \\
			2 & 1 \\
			1 & 2 \\
			1 & 2
		\end{pmatrix},
		\quad
		\text{and}
		\quad
		\bs{C}
		\coloneqq
		\begin{pmatrix}
			1 & 1 & 1 & 1 & 1 \\
			1 & 2 & 0 & 1 & 1 \\
			1 & 0 & 2 & 1 & 1 \\
			1 & 1 & 1 & 2 & 0 \\
			1 & 1 & 1 & 0 & 2
		\end{pmatrix}. \notag
	\end{align}
	\normalsize
Finally, we verify the remaining conditions of Assumption \ref{AJK}. First, $\textup{rank}( \bs{A} ) = 2$. Moreover, $\bs{\iota}_5 \notin \textup{col}( \bs{A} )$: if $\bs{\iota}_5 \in \textup{col}( \bs{A} )$ then there would exist $( a, b )^\top \in \mathbb{R}^2$ such that
\begin{align}
	( a + b,\, 2a + b,\, 2a + b,\, a + 2b,\, a + 2b )^\top
	=
	( 1, 1, 1, 1, 1 )^\top,
	\notag
\end{align}
which then implies $a + b = 1$ and $2a + b = 1$, hence $a = 0$ and $b = 1$, but then $a + 2b = 2 \neq 1$, which is a contradiction. Next, it is easily verified that $\bs{C}$ is positive semidefinite with (distinct) eigenvalues $0$, $2$, and $5$. Let $\bs{D} \coloneqq ( \bs{A}, \bs{\iota}_5 )$. Solving $\bs{u}^\top \bs{A} = \bs{0}_2^\top$ and $\bs{u}^\top \bs{\iota}_5 = 0$ yields
\begin{align}
	\textup{null}\big( \bs{D}^\top \big)
	=
	\textup{span}\big(
	( 0, -1, 1, 0, 0 )^\top,\,
	( 0, 0, 0, -1, 1 )^\top
	\big).
	\notag
\end{align}
For $\bs{u}_1 \coloneqq ( 0, -1, 1, 0, 0 )^\top$ and $\bs{u}_2 \coloneqq ( 0, 0, 0, -1, 1 )^\top$ we have $\bs{u}_1^\top \bs{C} \bs{u}_1 = \bs{u}_2^\top \bs{C} \bs{u}_2 = 4 > 0$, whereby $\textup{null}( \bs{D}^\top ) \not\subseteq \textup{null}( \bs{C} )$.
\end{exampleagainhr}

\subsection{Minimum-Variance Unbiased Jackknife Estimator}\label{MVUJ}
We now discuss how to combine a set of subsample estimators and the full sample estimator to produce a minimum-variance unbiased jackknife (MVUJ) estimator. Consider a vector $\bs{v} \in \mathbb{R}^{m}$ such that
\begin{align}
	\bs{v}^\top \bs{A}
	&=
	\bs{0}_{R}^\top,
	\label{v1} \\
	\bs{v}^\top \bs{\iota}_m
	&=
	1.
	\label{v2}
\end{align}
Then under Assumption \ref{AJK},
\begin{align}
	r_{NT} \big( \bs{v}^\top \bs{\hat{\varphi}} - \varphi \big)
	\xrightarrow{d}
	\mathcal{N} \big( 0, \sigma^2 \bs{v}^\top \bs{C} \bs{v} \big), \notag 
\end{align}
that is, $\bs{v}^\top \bs{\hat{\varphi}}$ is asymptotically unbiased. In general, the restrictions \eqref{v1} - \eqref{v2} do not identify a unique $\bs{v}$. Indeed, define
\begin{align}
	\bs{D}
	\coloneqq
	\begin{pmatrix}
		\bs{A} & \bs{\iota}_{m}
	\end{pmatrix}
	\quad
	\text{and}
	\quad
	\bs{d}
	\coloneqq
	\begin{pmatrix}
		\bs{0}_{R}\\[0.3em]
		1
	\end{pmatrix},
	\notag
\end{align}
so that \eqref{v1} and \eqref{v2} can be restated as $\bs{D}^\top \bs{v} = \bs{d}$. Since $\textup{rank}( \bs{A} ) = R$ and $\bs{\iota}_m \notin \textup{col}( \bs{A} )$, we have $\textup{rank}( \bs{D} ) = R + 1$, and hence the feasible set
\begin{align}
	\mathcal{V}
	=
	\big\{ \bs{v} \in \mathbb{R}^{m} : \bs{D}^\top \bs{v} = \bs{d} \big\},
	\notag
\end{align}
is nonempty whenever $m \geq R + 1$, and has infinitely many elements when $m \geq R + 2$. One can obtain an asymptotically unbiased jackknife estimator by solving the linear system $\bs{D}^\top \bs{v} = \bs{d}$ without reference to $\bs{C}$. For example, the minimum Euclidean norm solution is given by
\begin{align}
	\bs{v}^{\dagger}
	\coloneqq
	\bs{D} \big( \bs{D}^\top \bs{D} \big)^{-1} \bs{d}.
	\footnotemark
	\notag
\end{align}
Yet, since $\bs{C}$ is known, we can go further and construct an unbiased estimator with minimal asymptotic variance. In particular, an MVUJ weights vector is obtained by solving
\footnotetext{Since $\textup{rank}( \bs{D} ) = R + 1$ (because $\textup{rank}( \bs{A} ) = R$ and $\bs{\iota}_m \notin \textup{col}( \bs{A} )$), $\bs{D}$ has full column rank, hence $\bs{D}^\top \bs{D}$ is invertible.}
\begin{align}
	\bs{v}^\ast
	\in
	\arg\min_{ \bs{v} \in \mathcal{V} }
	\bs{v}^\top \bs{C} \bs{v}.
	\label{prog}
\end{align}
\begin{theorem}\label{thm:mvuj-existence}
	Under Assumptions \ref{AAD} and \ref{AJK}, the programme \eqref{prog} admits at least one solution $\bs{v}^\ast \in \mathcal{V}$.
\end{theorem}
\begin{proof}
	{\itshape
		We first show that $\mathcal{V}$ is nonempty. Since $\textup{rank}( \bs{A} ) = R$, we have $\dim( \textup{null}( \bs{A}^\top ) ) = m - R > 0$. The condition $\bs{\iota}_m \notin \textup{col}( \bs{A} )$ implies there exists $\bs{w} \in \textup{null}( \bs{A}^\top )$ such that $\bs{w}^\top \bs{\iota}_m \neq 0$. Setting $\bs{v}_0 \coloneqq	\frac{ \bs{w} }{ \bs{w}^\top \bs{\iota}_m }$	yields $\bs{v}_0 \in \mathcal{V}$. Next, since $\mathcal{V}$ forms a nonempty polyhedron, and $\bs{v}^\top \bs{C} \bs{v} \geq 0$ because $\bs{C} \succeq \bs{0}_{m \times m}$, it follows by the Frank--Wolfe Theorem that $\bs{v}^\top \bs{C} \bs{v}$ attains its infimum on $\mathcal{V}$ (see, e.g., \cite{cottle2009linear} Theorem 2.8.1). Hence, \eqref{prog} admits at least one minimiser.}
\end{proof}
Using a solution to \eqref{prog}, we define an MVUJ estimator as
\begin{align}
	\tilde{\varphi}
	\coloneqq
	\bs{v}^{\ast\top} \bs{\hat{\varphi}}.
	\notag
\end{align}
When $\bs{C}$ is positive definite, the programme \eqref{prog} has a unique solution which is given by
\begin{align}
	\bs{v}^\ast
	=
	\bs{C}^{-1} \bs{D}\big( \bs{D}^\top \bs{C}^{-1} \bs{D} \big)^{-1} \bs{d}.
	\label{sol-pd}
\end{align}
If $\bs{C}$ is not positive definite, the minimiser of \eqref{prog} need not be unique. In this case, one may compute an MVUJ weights vector by solving the system
\begin{align}
	\begin{pmatrix}
		2 \bs{C}    & \bs{D} \\
		\bs{D}^\top & \bs{0}_{R+1 \times R+1}
	\end{pmatrix}
	\begin{pmatrix}
		\bs{v} \\
		\bs{\pi}
	\end{pmatrix}
	=
	\begin{pmatrix}
		\bs{0}_m \\
		\bs{d}
	\end{pmatrix},
	\notag
\end{align}
where $\bs{\pi}$ is an $(R + 1) \times 1$ vector of Lagrange multipliers.\footnote{A necessary condition for $\bs{v}$ to be a solution to \eqref{prog} is that there exists a vector of Lagrange multipliers $\bs{\pi}$ such that $\begin{psmallmatrix} 2 \bs{C} & \bs{D} \\ \bs{D}^\top & \bs{0}_{R+1 \times R+1} \end{psmallmatrix} \binom{\bs{v}}{\bs{\pi}} = \binom{\bs{0}_m}{\bs{d}}$; see, e.g., \cite{nocedal2006numerical} Lemma 16.1 and Theorem 16.2. When $\bs{C} \succ \bs{0}$, this system has a unique solution $( \bs{v}^\ast, \bs{\pi}^\ast )$. Eliminating $\bs{\pi}^\ast$ yields \eqref{sol-pd}.} We now return to our examples.

\begin{exampleagainhr}[One-way Effects]{ex:onefe} \small Recall
	\begin{align}
		\bs{A}
		=
		\begin{pmatrix}
			1 \\
			2 \\
			2 
		\end{pmatrix}
		\qquad \text{and} \qquad
		\bs{C}
		=
		\begin{pmatrix}
			1 & 1 & 1\\
			1 & 2 & 0\\
			1 & 0 & 2
		\end{pmatrix}. \notag
	\end{align}
	An MVUJ weights vector solves
	\begin{align}
		\min_{ \bs{v} \in \mathcal{V} }
		\bs{v}^\top \bs{C} \bs{v} \qquad \text{where} \qquad 	\mathcal{V}
		=
		\left\{ \bs{v} \in \mathbb{R}^{3} : \bs{D}^\top \bs{v} = \bs{d} \right\}, \notag 
	\end{align}
	with
	\begin{align}
		\bs{D}
		\coloneqq
		\begin{pmatrix}
			\bs{A} & \bs{\iota}_3
		\end{pmatrix}
		\qquad \text{and} \qquad
		\bs{d}
		\coloneqq
		\begin{pmatrix}
			0 \\ 1
		\end{pmatrix}. \notag
	\end{align}
	The Karush–Kuhn–Tucker (KKT) system is
	\begin{align}
		\begin{pmatrix}
			2 \bs{C}    & \bs{D} \\
			\bs{D}^\top & \bs{0}_{2\times2}
		\end{pmatrix}
		\begin{pmatrix}
			\bs{v} \\ 
			\bs{\pi}
		\end{pmatrix}
		=
		\begin{pmatrix}
			\bs{0}_3 \\ \bs{d}
		\end{pmatrix}, \notag
	\end{align}
	where $\bs{\pi}$ is a $2 \times 1$ vector of Lagrange multipliers. Although $\bs{C}$ is singular, the left-hand matrix is nonsingular and hence the solution in $\bs{v}$ is unique. Direct calculation yields
	\begin{align}
		\bs{v}^\ast
		=
		\left( 2,\,-\frac{1}{2},\,-\frac{1}{2} \right)^\top. \notag
	\end{align}
	Therefore the MVUJ estimator is
	\begin{align}
		\tilde{\varphi}
		=
		\bs{v}^{\ast\top} \bs{\hat{\varphi}}
		=
		2 \hat{\varphi}^{(0)}
		-\frac{1}{2} \hat{\varphi}^{(1)}
		-\frac{1}{2} \hat{\varphi}^{(2)}. \notag
	\end{align}
	Moreover,
	\begin{align}
		r_{NT} \big( \tilde{\varphi} - \varphi \big)
		\xrightarrow{d}
		\mathcal{N}\big( 0,\,\sigma^2 \big). \notag
	\end{align}
\end{exampleagainhr}

\begin{exampleagainhr}[Two-way Effects]{ex:twofe} \small
	Recall
	\begin{align}
		\bs{A}
		=
		\begin{pmatrix}
			1 & 1 \\
			2 & 1 \\
			2 & 1 \\
			1 & 2 \\
			1 & 2
		\end{pmatrix}
		\qquad \text{and} \qquad
		\bs{C}
		=
		\begin{pmatrix}
			1 & 1 & 1 & 1 & 1 \\
			1 & 2 & 0 & 1 & 1 \\
			1 & 0 & 2 & 1 & 1 \\
			1 & 1 & 1 & 2 & 0 \\
			1 & 1 & 1 & 0 & 2
		\end{pmatrix}.
		\notag
	\end{align}
	An MVUJ weights vector solves
	\begin{align}
		\min_{ \bs{v} \in \mathcal{V} }
		\bs{v}^\top \bs{C} \bs{v} 
		\qquad
		\text{where} 
		\qquad
		\mathcal{V}
		=
		\left\{ \bs{v} \in \mathbb{R}^{5} : \bs{D}^\top \bs{v} = \bs{d} \right\},
		\notag
	\end{align}
	with
	\begin{align}
		\bs{D}
		\coloneqq
		\begin{pmatrix}
			\bs{A} & \bs{\iota}_5
		\end{pmatrix}
		\quad \text{and} \quad
		\bs{d}
		\coloneqq
		\begin{pmatrix}
			\bs{0}_2 \\
			1
		\end{pmatrix}.
		\notag
	\end{align}
	The KKT system is
	\begin{align}
		\begin{pmatrix}
			2 \bs{C} & \bs{D} \\
			\bs{D}^{\top} & \bs{0}_{3\times3}
		\end{pmatrix}
		\begin{pmatrix}
			\bs{v} \\
			\bs{\pi}
		\end{pmatrix}
		=
		\begin{pmatrix}
			\bs{0}_5 \\
			\bs{d}
		\end{pmatrix},
		\notag
	\end{align}
	where $\bs{\pi}$ is a $3 \times 1$ vector of Lagrange multipliers. Although $\bs{C}$ is singular, the left-hand matrix is nonsingular and hence the solution in $\bs{v}$ is unique. Direct calculation yields
\begin{align}
	\bs{v}^\ast
	=
	\left( 3,\,-\frac{1}{2},\,-\frac{1}{2},\,-\frac{1}{2},\,-\frac{1}{2} \right)^\top.
	\notag
\end{align}
Therefore the MVUJ estimator is
\begin{align}
	\tilde{\varphi}
	=
	\bs{v}^{\ast\top} \bs{\hat{\varphi}}
	=
	3\,\hat{\varphi}^{(0)}
	-
	\frac{1}{2} \hat{\varphi}^{(1)}
	-
	\frac{1}{2} \hat{\varphi}^{(2)}
	-
	\frac{1}{2} \hat{\varphi}^{(3)}
	-
	\frac{1}{2} \hat{\varphi}^{(4)}.
	\notag
\end{align}
Moreover,
\begin{align}
	r_{NT} \big( \tilde{\varphi} - \varphi \big)
	\xrightarrow{d}
	\mathcal{N}\big( 0, \sigma^2 \big).
	\notag
\end{align}
\end{exampleagainhr}

\subsection{Jackknife $t$-Statistic}\label{jkt}
We now construct a jackknife $t$-statistic that delivers inference without knowledge of $\sigma^2$ (nor indeed $\bs{\mu}_{NT}$). Consider a vector $\bs{u}^\ast$ such that
\begin{align}
	\bs{u}^{\ast \top} \bs{A}
	&=
	\bs{0}_{R}^\top,
	\label{eq:u-bias} \\
	\bs{u}^{\ast\top} \bs{\iota}_m
	&=
	0,
	\label{eq:u-sum} \\
	\bs{u}^{\ast \top} \bs{C} \bs{u}^\ast
	&=
	\bs{v}^{\ast \top} \bs{C} \bs{v}^\ast,
	\label{eq:var-match}
\end{align}
where $\bs{A}$ and $\bs{C}$ are as in Assumption \ref{AJK}. We establish the existence of such a vector in the following result.

\begin{theorem}\label{thm:u-existence}
	Under Assumptions \ref{AAD} and \ref{AJK}, there exists at least one vector $\bs{u}^\ast$ satisfying \eqref{eq:u-bias} - \eqref{eq:var-match}. Moreover, $\bs{u}^{\ast \top} \bs{C} \bs{v}^\ast = 0$.
\end{theorem}

\begin{proof}
	{\itshape
		Let $\bs{D} \coloneqq ( \bs{A}, \bs{\iota}_m )$. Since $\bs{v}^\ast$ solves \eqref{prog}, the KKT conditions imply there exists $\bs{\pi}$ such that $2 \bs{C} \bs{v}^\ast + \bs{D}\bs{\pi} = \bs{0}_m$, and hence $\bs{C} \bs{v}^\ast \in \textup{col}( \bs{D} )$. Therefore, for any $\bs{u}$ satisfying $\bs{u}^\top \bs{D} = \bs{0}_{R+1}^\top$, we have $\bs{u}^\top \bs{C} \bs{v}^\ast = 0$. Next, by Assumption \ref{AJK}, $\textup{null}( \bs{D}^\top ) \not\subseteq	\textup{null}( \bs{C} )$, so there exists $\bs{u}_0 \in \textup{null}\big( \bs{D}^\top \big)$ such that $\bs{u}_0^\top \bs{C} \bs{u}_0 > 0$. If $\bs{v}^{\ast \top} \bs{C} \bs{v}^\ast = 0$, set $\bs{u}^\ast \coloneqq \bs{0}_m$. Otherwise, define
		\begin{align}
			\bs{u}^\ast
			\coloneqq
			\left(
			\frac{ \bs{v}^{\ast \top} \bs{C} \bs{v}^\ast }{ \bs{u}_0^\top \bs{C} \bs{u}_0 }
			\right)^{ \frac{1}{2} }
			\bs{u}_0.
			\notag
		\end{align}
		Then $\bs{u}^{\ast\top} \bs{D} = \bs{0}_{R+1}^\top$, which is \eqref{eq:u-bias} - \eqref{eq:u-sum}. Moreover, $
		\bs{u}^{\ast\top} \bs{C} \bs{u}^\ast = \bs{v}^{\ast\top} \bs{C} \bs{v}^\ast$, so \eqref{eq:var-match} holds, and we have already established $\bs{u}^{\ast\top} \bs{C} \bs{v}^\ast = 0$.}
\end{proof}
Now that we have established the existence of vectors $\bs{v}^\ast$ and $\bs{u}^\ast$ which satisfy \eqref{v1}, \eqref{v2}, and \eqref{eq:u-bias} - \eqref{eq:var-match}, we use these to obtain the following result.

\begin{theorem}\label{thm:jt}
	Assume $\mathcal{V}\ \cap \ \textup{null}( \bs{C} ) = \emptyset$. Under Assumptions \ref{AAD} and \ref{AJK},
	\begin{align}
		\mathcal{J}
		\coloneqq
		\frac{ \tilde{\varphi} - \varphi }{ \tilde{\sigma} }
		\xrightarrow{d}
		t_1,
		\label{res}
	\end{align}
	as $N, T \rightarrow \infty$ and $N / T \rightarrow \gamma^2 \in ( 0, \infty )$, where
	\begin{align}
		\tilde{\sigma}^2
		\coloneqq
		\big( \bs{u}^{\ast\top} \bs{\hat{\varphi}} \big)^2
		\qquad \text{and} \qquad
		\tilde{\sigma}
		\coloneqq
		\sqrt{ \tilde{\sigma}^2 }.
		\notag
	\end{align}
\end{theorem}

\begin{proof}
	{\itshape
		By Assumption \ref{AJK} and the continuous mapping theorem,
		\begin{align}
			\begin{pmatrix}
				r_{NT} \bs{v}^{\ast \top} \big( \bs{\hat{\varphi}} - \varphi \bs{\iota}_m \big) \\
				r_{NT} \bs{u}^{\ast \top} \big( \bs{\hat{\varphi}} - \varphi \bs{\iota}_m \big)
			\end{pmatrix}
			\xrightarrow{d}
			\mathcal{N} \left( \bs{0}_2, \sigma^2
			\begin{pmatrix}
				\bs{v}^{\ast\top} \bs{C} \bs{v}^\ast & \bs{v}^{\ast\top} \bs{C} \bs{u}^\ast \\
				\bs{u}^{\ast\top} \bs{C} \bs{v}^\ast & \bs{u}^{\ast\top} \bs{C} \bs{u}^\ast
			\end{pmatrix}
			\right).
			\notag
		\end{align}
		By Theorem \ref{thm:u-existence}, $\bs{u}^{\ast\top} \bs{C} \bs{v}^\ast = 0$ and $\bs{u}^{\ast\top} \bs{C} \bs{u}^\ast = \bs{v}^{\ast\top} \bs{C} \bs{v}^\ast$. Therefore the limiting covariance matrix equals $\sigma^2 \bs{v}^{\ast\top} \bs{C} \bs{v}^\ast \bs{I}_2 \succ \bs{0}_{2 \times 2}$ which follows since $\mathcal{V}\ \cap \ \textup{null}( \bs{C} ) = \emptyset$. In particular, the two components are asymptotically independent and identically distributed. Since $\bs{v}^{\ast \top} \bs{\iota}_m = 1$ we have $\bs{v}^{\ast \top} \big( \bs{\hat{\varphi}} - \varphi \bs{\iota}_m \big) = \bs{v}^{\ast \top} \bs{\hat{\varphi}} - \varphi$, and since $\bs{u}^{\ast \top} \bs{\iota}_m = 0$, we have $\bs{u}^{\ast \top} \big( \bs{\hat{\varphi}} - \varphi \bs{\iota}_m \big) = \bs{u}^{\ast \top} \bs{\hat{\varphi}}$. Hence
		\begin{align}
			\frac{ r_{NT} \big( \bs{v}^{\ast\top} \bs{\hat{\varphi}} - \varphi \big) }{ r_{NT} \left| \bs{u}^{\ast \top} \bs{\hat{\varphi}} \right| }
			=
			\frac{ \tilde{\varphi} - \varphi }{ \tilde{\sigma} }
			\xrightarrow{d}
			t_1.
			\notag
	\end{align}}
\end{proof}

We refer to $\mathcal{J}$ as the \textit{jackknife $t$-statistic}, and to $\tilde{\sigma}^2$ and $\tilde{\sigma}$ as the \textit{jackknife variance} and \textit{jackknife standard error}, respectively. The jackknife $t$-statistic provides a straightforward way to conduct inference on $\varphi$. In particular, by constructing $\bs{\hat{\varphi}}$ and then calculating $\bs{u}^\ast$ and $\bs{v}^\ast$ we can readily obtain hypothesis tests, $p$-values, and confidence intervals, as discussed below.\\

\noindent \textbf{Hypothesis tests.} To test
\begin{align}
	H_0: \varphi = \varphi_0
	\qquad \text{against} \qquad
	H_A: \varphi \neq \varphi_0, \notag 
\end{align}
we form the jackknife $t$-statistic
\begin{align}
	\mathcal{J}
	\coloneqq
	\frac{ \tilde{\varphi} - \varphi_0 }{\tilde{\sigma}}. \label{sched}
\end{align}
A level-$\alpha$ two-sided jackknife $t$-test rejects $H_0$ whenever $| \mathcal{J} | > t_{1,1-\alpha/2}$, where $t_{1,1-\alpha/2}$ is the $( 1 - \alpha/2 )$ quantile of a $t$-distribution with $1$ degree of freedom. One-sided tests are obtained similarly. \\

\noindent \textbf{$p$-values.} For the two-sided test of $H_0: \varphi = \varphi_0$ against $H_A: \varphi \neq \varphi_0$, the $p$-value associated with the test statistic \eqref{sched} is
	\begin{align}
		p
		\coloneqq
		2 
		\left(
		1 - F\left( \left| \mathcal{J} \right| \right)
		\right), \notag
	\end{align}
	where $F( \cdot )$ denotes the CDF of the $t$-distribution with $1$ degree of freedom. For one-sided alternatives $H_A: \varphi > \varphi_0$ and $H_A: \varphi < \varphi_0$, the corresponding $p$-values are
	\begin{align}
		p^{+}
		&\coloneqq
		1 - F \left( \mathcal{J} \right)
		\qquad \text{for}\ H_A: \varphi > \varphi_0, \notag \\[0.3em]
		p^{-}
		&\coloneqq
		F \left( \mathcal{J} \right)
		\hspace{1.4cm}\text{for}\ H_A: \varphi < \varphi_0. \notag
	\end{align}

\noindent \textbf{Confidence intervals.} A $(1-\alpha)$ confidence interval for $\varphi$ is given by
	\begin{align}
		\left[ \,
		\tilde{\varphi} \,
		\pm \,
		t_{1,1-\alpha/2} \,
		\tilde{\sigma} \,
		\right], \notag
	\end{align}
	where $t_{1,1-\alpha/2}$ is the $( 1 - \alpha/2 )$ quantile of the $t$-distribution with $1$ degree of freedom. One-sided confidence intervals can be obtained similarly. We now return to our examples to illustrate our method of inference.

\begin{exampleagainhr}[One-way Effects]{ex:onefe}
	Recall 
	\begin{align}
		\bs{\hat{\varphi}}
		=
		\begin{pmatrix}
			\hat{\varphi}^{(0)} \\
			\hat{\varphi}^{(1)} \\
			\hat{\varphi}^{(2)}
		\end{pmatrix},
		\quad
		\bs{A} 
		=
		\begin{pmatrix}
			1 \\[0.1em]
			2 \\[0.1em]
			2
		\end{pmatrix},
		\quad
		\text{and}
		\quad
		\bs{C}
		=
		\begin{pmatrix}
			1 & 1 & 1 \\
			1 & 2 & 0 \\
			1 & 0 & 2
		\end{pmatrix}. \notag 
	\end{align}
	The MVUJ weights vector is 
	\begin{align}
		\bs{v}^\ast
		=
		\left( 2, \, -\frac{1}{2}, \, -\frac{1}{2} \right)^\top, \notag 
	\end{align}
	so that
	\begin{align}
		\tilde{\varphi}
		=
		\bs{v}^{\ast \top} \bs{\hat{\varphi}}
		=
		2 \hat{\varphi}^{(0)}
		- \frac{1}{2} \hat{\varphi}^{(1)}
		- \frac{1}{2} \hat{\varphi}^{(2)}. \notag 
	\end{align}	
	We now construct $\bs{u}^\ast$ satisfying \eqref{eq:u-bias} - \eqref{eq:var-match}.
	First note that $\bs{C} \bs{v}^\ast = \bs{\iota}_3$, so the condition $\bs{u}^{\ast\top} \bs{C} \bs{v}^\ast = 0$ becomes $\bs{u}^{\ast \top} \bs{\iota}_3 = 0$, which coincides with \eqref{eq:u-sum}. Moreover, since $\bs{A}$ is a scalar multiple of $( 1, 2, 2 )^\top$, the restriction \eqref{eq:u-bias} is equivalent to
\begin{align}
	u_1 + 2u_2 + 2u_3 = 0.
	\notag
\end{align}
Imposing $u_1 + u_2 + u_3 = 0$ and $u_1 + 2u_2 + 2u_3 = 0$ yields $u_2 = - u_3$ and $u_1 = 0$. A convenient choice is
\begin{align}
	\bs{u}^\ast
	\coloneqq
	\left( 0,\, \frac{1}{2},\, -\frac{1}{2} \right)^\top.
	\notag
\end{align}
For this choice,
\begin{align}
	\bs{u}^{\ast \top} \bs{A} = 0,
	\qquad
	\bs{u}^{\ast \top} \bs{\iota}_3 = 0,
	\qquad
	\bs{u}^{\ast \top} \bs{C} \bs{v}^\ast = 0,
	\quad \text{and} \quad
	\bs{u}^{\ast \top} \bs{C} \bs{u}^\ast = \bs{v}^{\ast \top} \bs{C} \bs{v}^\ast = 1,
	\notag
\end{align}
so \eqref{eq:u-bias} - \eqref{eq:var-match} hold. The jackknife standard error is
\begin{align}
	\tilde{\sigma}
	=
	\left| \bs{u}^{\ast \top} \bs{\hat{\varphi}} \right|
	=
	\frac{1}{2}
	\left|
	\hat{\varphi}^{(1)} - \hat{\varphi}^{(2)}
	\right|,
	\notag
\end{align}
and hence the jackknife $t$-statistic is given by
\begin{align}
	\mathcal{J}
	=
	\frac{ \tilde{\varphi} - \varphi }{ \tilde{\sigma} }
	=
	\frac{ 4\hat{\varphi}^{(0)} - \hat{\varphi}^{(1)} - \hat{\varphi}^{(2)} - 2\varphi }{ \left| \hat{\varphi}^{(1)} - \hat{\varphi}^{(2)} \right| }
	\xrightarrow{d}
	t_1.
	\notag
\end{align}
A corresponding $( 1 - \alpha )$ two-sided confidence interval is obtained as
\begin{align}
	\left[
	2 \hat{\varphi}^{(0)}
	-
	\frac{1}{2}\hat{\varphi}^{(1)}
	-
	\frac{1}{2}\hat{\varphi}^{(2)}
	\ \pm\
	t_{1,1-\alpha/2}\,
	\frac{1}{2}
	\left|
	\hat{\varphi}^{(1)} - \hat{\varphi}^{(2)}
	\right|
	\right].
	\notag
\end{align}
\end{exampleagainhr}

\begin{exampleagainhr}[Two-way Effects]{ex:twofe} \small
	Recall
	\begin{align}
		\bs{\hat{\varphi}}
		=
		\begin{pmatrix}
			\hat{\varphi}^{(0)} \\
			\hat{\varphi}^{(1)} \\
			\hat{\varphi}^{(2)} \\
			\hat{\varphi}^{(3)} \\
			\hat{\varphi}^{(4)}
		\end{pmatrix},
		\quad
		\bs{A}
		=
		\begin{pmatrix}
			1 & 1 \\
			2 & 1 \\
			2 & 1 \\
			1 & 2 \\
			1 & 2
		\end{pmatrix},
		\quad
		\text{and}
		\quad
		\bs{C}
		=
		\begin{pmatrix}
			1 & 1 & 1 & 1 & 1 \\
			1 & 2 & 0 & 1 & 1 \\
			1 & 0 & 2 & 1 & 1 \\
			1 & 1 & 1 & 2 & 0 \\
			1 & 1 & 1 & 0 & 2
		\end{pmatrix}.
		\notag
	\end{align}
	The MVUJ weights vector is
	\begin{align}
		\bs{v}^\ast
		=
		\left( 3,\, -\frac{1}{2},\, -\frac{1}{2},\, -\frac{1}{2},\, -\frac{1}{2} \right)^\top,
		\notag
	\end{align}
	so that
	\begin{align}
		\tilde{\varphi}
		=
		\bs{v}^{\ast \top} \bs{\hat{\varphi}}
		=
		3 \hat{\varphi}^{(0)}
		-
		\frac{1}{2} \hat{\varphi}^{(1)}
		-
		\frac{1}{2} \hat{\varphi}^{(2)}
		-
		\frac{1}{2} \hat{\varphi}^{(3)}
		-
		\frac{1}{2} \hat{\varphi}^{(4)}.
		\notag
	\end{align}
	We now construct $\bs{u}^\ast$ satisfying \eqref{eq:u-bias} - \eqref{eq:var-match}. First note that $\bs{C} \bs{v}^\ast = \bs{\iota}_5$, so the condition $\bs{u}^{\ast \top} \bs{C} \bs{v}^\ast = 0$ becomes $\bs{u}^{\ast \top} \bs{\iota}_5 = 0$, which coincides with \eqref{eq:u-sum}. Moreover, since $\bs{A}$ has columns proportional to $( 1, 2, 2, 1, 1 )^\top$ and $( 1, 1, 1, 2, 2 )^\top$, the restriction \eqref{eq:u-bias} is equivalent to
	\begin{align}
		u_1 + 2u_2 + 2u_3 + u_4 + u_5
		&=
		0,
		\notag \\
		u_1 + u_2 + u_3 + 2u_4 + 2u_5
		&=
		0.
		\notag
	\end{align}
	Imposing $u_1 + u_2 + u_3 + u_4 + u_5 = 0$ and the two restrictions above yields $u_1 = 0$, $u_2 = -u_3$, and $u_4 = -u_5$. A convenient choice is
	\begin{align}
		\bs{u}^\ast
		\coloneqq
		\left( 0,\, -\frac{1}{2},\, \frac{1}{2},\, 0,\, 0 \right)^\top.
		\notag
	\end{align}
	For this choice,
	\begin{align}
		\bs{u}^{\ast \top} \bs{A}
		=
		\bs{0}_2^\top,
		\qquad
		\bs{u}^{\ast \top} \bs{\iota}_5
		=
		0,
		\qquad
		\bs{u}^{\ast \top} \bs{C} \bs{v}^\ast
		=
		0,
		\quad \text{and} \quad
		\bs{u}^{\ast \top} \bs{C} \bs{u}^\ast
		=
		\bs{v}^{\ast \top} \bs{C} \bs{v}^\ast
		=
		1.
		\notag
	\end{align}
	Thus \eqref{eq:u-bias} - \eqref{eq:var-match} hold. The jackknife standard error is
	\begin{align}
		\tilde{\sigma}
		=
		\left| \bs{u}^{\ast \top} \bs{\hat{\varphi}} \right|
		=
		\frac{1}{2}
		\left|
		\hat{\varphi}^{(2)} - \hat{\varphi}^{(1)}
		\right|,
		\notag
	\end{align}
	and hence the jackknife $t$-statistic is given by
	\begin{align}
		\mathcal{J}
		\coloneqq
		\frac{
			6 \hat{\varphi}^{(0)}
			-
			\hat{\varphi}^{(1)}
			-
			\hat{\varphi}^{(2)}
			-
			\hat{\varphi}^{(3)}
			-
			\hat{\varphi}^{(4)}
			-
			2\varphi
		}{
			\left|
			\hat{\varphi}^{(2)} - \hat{\varphi}^{(1)}
			\right|
		}
		\xrightarrow{d}
		t_1.
		\notag
	\end{align}
	A corresponding $( 1 - \alpha )$ two-sided confidence interval is obtained as
	\begin{align}
		\left[
		3 \hat{\varphi}^{(0)}
		-
		\frac{1}{2}\hat{\varphi}^{(1)}
		-
		\frac{1}{2}\hat{\varphi}^{(2)}
		-
		\frac{1}{2}\hat{\varphi}^{(3)}
		-
		\frac{1}{2}\hat{\varphi}^{(4)}
		\ \pm\
		t_{1,1-\alpha/2}\,
		\frac{1}{2}
		\left|
		\hat{\varphi}^{(2)} - \hat{\varphi}^{(1)}
		\right|
		\right].
		\notag
	\end{align}
\end{exampleagainhr}

\subsection{Jackknife $t_q$-Statistic}\label{jktq}
Section \ref{jkt} describes the simplest construction, in which we find only a single variance weights vector $\bs{u}^\ast$ and obtain a jackknife $t$-statistic which is asymptotically $t$-distributed with one degree of freedom. However, in many settings it is possible to find multiple variance weights vectors $\bs{u}_1^\ast, \ldots, \bs{u}_q^\ast$ satisfying \eqref{eq:u-bias} - \eqref{eq:var-match} through the generation of further subsamples. These can be utilised to construct a statistic which is $t_q$-distributed with $q > 1$. This has the advantage that the critical values $t_{q,1-\alpha/2}$ are smaller than $t_{1,1-\alpha/2}$ and therefore produce shorter confidence intervals, for example. We now formalise this. Suppose there exist $q$ vectors $\bs{u}_1^\ast, \ldots, \bs{u}_q^\ast$ such that, writing
\begin{align}
	\bs{U}^\ast
	\coloneqq
	\big( \bs{u}_1^\ast, \ldots, \bs{u}_q^\ast \big),
	\notag
\end{align}
the following conditions hold:
\begin{align}
	\bs{U}^{\ast \top} \bs{A}
	&=
	\bs{0}_{q \times R},
	\label{eq:uq-bias} \\
	\bs{U}^{\ast \top} \bs{\iota}_m
	&=
	\bs{0}_{q},
	\label{eq:uq-sum} \\
	\bs{U}^{\ast \top} \bs{C} \bs{U}^\ast
	&=
	\bs{v}^{\ast \top} \bs{C} \bs{v}^\ast \bs{I}_q.
	\label{eq:uq-var-match}
\end{align}
These conditions are natural extensions of \eqref{eq:u-bias} - \eqref{eq:var-match}. Using $\bs{v}^\ast$ and $\bs{U}^\ast$ we obtain the following result.

\begin{theorem}\label{thm:jtq}
	Assume $\mathcal{V}\ \cap \ \textup{null}( \bs{C} ) = \emptyset$. Let $\bs{v}^\ast$ solve \eqref{prog}, and let $\bs{u}_1^\ast, \ldots, \bs{u}_q^\ast$ be such that with $\bs{U}^\ast \coloneqq ( \bs{u}_1^\ast, \ldots, \bs{u}_q^\ast )$ conditions \eqref{eq:uq-bias} - \eqref{eq:uq-var-match} hold. Under Assumptions \ref{AAD} and \ref{AJK},
	\begin{align}
		\mathcal{J}_q
		\coloneqq
		\frac{ \tilde{\varphi} - \varphi }{ \tilde{\sigma}_q }
		\xrightarrow{d}
		t_q,
		\label{resq}
	\end{align}
	as $N, T \rightarrow \infty$ and $N / T \rightarrow \gamma^2 \in ( 0, \infty )$, where
	\begin{align}
		\tilde{\sigma}_q^2
		\coloneqq
		\frac{1}{q}
		\sum_{l=1}^q
		\big( \bs{u}_l^{\ast\top} \bs{\hat{\varphi}} \big)^2,
		\qquad \text{and} \qquad
		\tilde{\sigma}_q
		\coloneqq
		\sqrt{ \tilde{\sigma}_q^2 }.
		\notag
	\end{align}
\end{theorem}

\begin{proof}
	{\itshape
		Let $\bs{D} \coloneqq ( \bs{A}, \bs{\iota}_m )$. Since $\bs{v}^\ast$ solves \eqref{prog}, the KKT conditions imply there exists an $( R + 1 ) \times 1$ vector $\bs{\pi}$ such that $2 \bs{C} \bs{v}^\ast + \bs{D}\bs{\pi} = \bs{0}_m$, and hence $\bs{C} \bs{v}^\ast \in \textup{col}( \bs{D} )$. For each $l$, conditions \eqref{eq:uq-bias} - \eqref{eq:uq-sum} imply $\bs{u}_l^{\ast\top} \bs{D} = \bs{0}_{R+1}^\top$, so $\bs{u}_l^{\ast\top} \bs{C} \bs{v}^\ast = 0$. By Assumption \ref{AJK} and the continuous mapping theorem,
		\begin{align}
			\begin{pmatrix}
				r_{NT} \big( \tilde{\varphi} - \varphi \big) \\
				r_{NT} \bs{U}^{\ast \top}\big( \bs{\hat{\varphi}} - \varphi \bs{\iota}_m \big)
			\end{pmatrix}
			\xrightarrow{d}
			\mathcal{N}\left(
			\bs{0}_{q+1},
			\sigma^2
			\begin{pmatrix}
				\bs{v}^{\ast \top} \bs{C} \bs{v}^\ast & \bs{v}^{\ast \top} \bs{C} \bs{U}^\ast \\
				\bs{U}^{\ast \top} \bs{C} \bs{v}^\ast & \bs{U}^{\ast \top} \bs{C} \bs{U}^\ast
			\end{pmatrix}
			\right).
			\notag
		\end{align}
		By the previous paragraph, $\bs{U}^{\ast \top} \bs{C} \bs{v}^\ast = \bs{0}_{q \times 1}$, and by \eqref{eq:uq-var-match} we have $\bs{U}^{\ast \top} \bs{C} \bs{U}^\ast = \bs{v}^{\ast \top} \bs{C} \bs{v}^\ast \bs{I}_q$. Therefore the limiting covariance matrix equals $\sigma^2 \bs{v}^{\ast\top} \bs{C} \bs{v}^\ast \bs{I}_{q+1}$. Since $\mathcal{V}\ \cap \ \textup{null}( \bs{C} ) = \emptyset$, we have $\bs{v}^{\ast \top} \bs{C} \bs{v}^\ast > 0$ and hence $\sigma^2 \bs{v}^{\ast \top} \bs{C} \bs{v}^\ast > 0$. Moreover, \eqref{eq:uq-sum} implies $\bs{U}^{\ast \top}( \bs{\hat{\varphi}} - \varphi \bs{\iota}_m ) = \bs{U}^{\ast \top} \bs{\hat{\varphi}}$, so the last $q$ components above are exactly $r_{NT}\bs{u}_l^{\ast \top} \bs{\hat{\varphi}}$. It follows that
		\begin{align}
			\frac{1}{ \sigma^2 \bs{v}^{\ast \top} \bs{C} \bs{v}^\ast }
			\sum_{l=1}^q
			\big( r_{NT}\bs{u}_l^{\ast \top} \bs{\hat{\varphi}} \big)^2
			\xrightarrow{d}
			\chi^2_q,
			\notag
		\end{align}
		and hence
		\begin{align}
			\frac{ \tilde{\sigma}_q^2 }{ \sigma^2 \bs{v}^{\ast\top} \bs{C} \bs{v}^\ast / r_{NT}^2 }
			=
			\frac{1}{q}
			\sum_{l=1}^q
			\left(
			\frac{ r_{NT}\bs{u}_l^{\ast\top} \bs{\hat{\varphi}}	}{ \sqrt{ \sigma^2 \bs{v}^{\ast\top} \bs{C} \bs{v}^\ast }}
			\right)^2
			\xrightarrow{d}
			\frac{1}{q}\chi^2_q.
			\notag
		\end{align}
		Therefore,
\begin{align}
	\frac{ r_{NT} \big( \tilde{\varphi} - \varphi \big) }{ r_{NT}\tilde{\sigma}_q }
	=
	\frac{ \tilde{\varphi} - \varphi }{ \tilde{\sigma}_q } 
	=
	\frac{ \frac{ r_{NT} \big( \tilde{\varphi} - \varphi \big) }{ \sqrt{ \sigma^2 \bs{v}^{\ast\top} \bs{C} \bs{v}^\ast } } }{
	\sqrt{ 	\frac{1}{q} \sum_{l=1}^q \left( \frac{ r_{NT}\bs{u}_l^{\ast\top} \bs{\hat{\varphi}} }{ \sqrt{ \sigma^2 \bs{v}^{\ast\top} \bs{C} \bs{v}^\ast } } \right)^2 	}
	} 
	\xrightarrow{d}
	t_q.
	\notag
\end{align}
		which yields \eqref{resq}.}
\end{proof}
We now return to our examples to illustrate the construction of a jackknife $t_q$-statistic with $q > 1$.

\begin{exampleagainhr}[One-way Effects]{ex:onefe}\small
	In our previous one-way fixed effects example we considered $m = 3$ estimators: the full sample estimator and two subsample estimators. In that design,
	\begin{align}
		\bs{A}
		=
		\begin{pmatrix}
			1 \\ 2 \\ 2
		\end{pmatrix},
		\notag
	\end{align}
	so $m -	\textup{rank}( ( \bs{A}, \bs{\iota}_3 ) ) = 1$,	and hence there is (up to scale) at most one admissible $\bs{u}^\ast$ and the resulting jackknife $t$-statistic is necessarily $t_1$-based. To illustrate the case $q > 1$, we consider a subsampling scheme that partitions the time dimension into nonoverlapping thirds. To be specific, assume $T = 3T_3$ where $T_3$ is an integer. Consider subsamples
	\begin{align}
		S_1
		&\coloneqq
		\{ 1, \ldots, N \} \times \{ 1, \ldots, T_3 \},
		\notag \\
		S_2
		&\coloneqq
		\{ 1, \ldots, N \} \times \{ T_3 + 1, \ldots, 2T_3 \},
		\notag \\
		S_3
		&\coloneqq
		\{ 1, \ldots, N \} \times \{ 2T_3 + 1, \ldots, T \}.
		\notag
	\end{align}
	Let $\hat{\varphi}^{(j)}$ be the LS estimator using subsample $S_j$ for $j = 1, 2, 3$. Collecting these with the full-sample estimator $\hat{\varphi}^{(0)}$ gives
	\begin{align}
		\bs{\hat{\varphi}}
		=
		\big(  \,
		\hat{\varphi}^{(0)}, \,
		\hat{\varphi}^{(1)}, \,
		\hat{\varphi}^{(2)}, \,
		\hat{\varphi}^{(3)}  \,
		\big)^\top.
		\notag
	\end{align}
	As previously we obtain,
	\begin{align}
		\sqrt{NT} \big( \bs{\hat{\varphi}} - \varphi \bs{\iota}_4 \big)
		=
		\bs{z}_{NT}
		+
		\bs{A} \bs{\mu}_{NT}
		+
		\smallO_{p}( 1 ),
		\notag
	\end{align}
	with $\bs{z}_{NT} \xrightarrow{d} \mathcal{N}( \bs{0}_4, \sigma^2 \bs{C} )$,
	\begin{align}
		\bs{A}
		=
		\begin{pmatrix}
			1 \\[0.1em]
			3 \\[0.1em]
			3 \\[0.1em]
			3
		\end{pmatrix},
		\qquad \text{and} \qquad
		\bs{C}
		=
		\begin{pmatrix}
			1 & 1 & 1 & 1 \\
			1 & 3 & 0 & 0 \\
			1 & 0 & 3 & 0 \\
			1 & 0 & 0 & 3
		\end{pmatrix}.
		\notag
	\end{align}
	Moreover, $m - \textup{rank}( ( \bs{A}, \bs{\iota}_4 ) ) = 2$, so there is (up to scale) at most two admissible variance weights vectors. There is a unique MVUJ estimator associated with this design which is given by
	\begin{align}
		\tilde{\varphi}
		=
		\frac{3}{2}\hat{\varphi}^{(0)}
		-
		\frac{1}{6}\hat{\varphi}^{(1)}
		-
		\frac{1}{6}\hat{\varphi}^{(2)}
		-
		\frac{1}{6}\hat{\varphi}^{(3)}.
		\notag
	\end{align}
	Notice that $\bs{C}\bs{v}^\ast = \bs{\iota}_4$ and so the restriction $\bs{u}^\top \bs{C}\bs{v}^\ast = 0$ coincides with $\bs{u}^\top \bs{\iota}_4 = 0$. Thus, we may take two linearly independent variance weights vectors and then normalise them to satisfy \eqref{eq:uq-var-match}. For instance, define
	\begin{align}
		\bs{u}_1^\ast
		\coloneqq
		\left( 0,\,-\frac{1}{\sqrt{6}},\,\frac{1}{\sqrt{6}},\,0 \right)^\top
		\qquad \text{and} \qquad
		\bs{u}_2^\ast
		\coloneqq
		\left( 0,\,-\sqrt{ \frac{1}{18} },\,-\sqrt{ \frac{1}{18} },\,\sqrt{ \frac{2}{9} } \right)^\top.
		\notag
	\end{align}
	Then writing $\bs{U}^\ast \coloneqq ( \bs{u}_1^\ast, \bs{u}_2^\ast )$,
	\begin{align}
		\bs{U}^{\ast \top} \bs{A} = \bs{0}_2^\top,
		\qquad
		\bs{U}^{\ast \top} \bs{\iota}_4 = \bs{0}_2,
		\quad \text{and} \quad
		\bs{U}^{\ast \top}\bs{C}\bs{U}^\ast = \bs{v}^{\ast \top} \bs{C} \bs{v}^\ast \bs{I}_2.
		\notag
	\end{align}
	The corresponding jackknife standard error is
	\begin{align}
		\tilde{\sigma}_2^2
		=
		\frac{1}{2}\big( \bs{u}_1^{\ast \top} \bs{\hat{\varphi}} \big)^2
		+
		\frac{1}{2}\big( \bs{u}_2^{\ast \top} \bs{\hat{\varphi}} \big)^2,
		\notag
	\end{align}
	and the jackknife $t_2$-statistic is
	\begin{align}
		\mathcal{J}_2
		\coloneqq
		\frac{ \tilde{\varphi} - \varphi }{ \tilde{\sigma}_2 }
		\xrightarrow{d}
		t_2.
		\notag
	\end{align}
\end{exampleagainhr}

\begin{exampleagainhr}[Two-way Effects]{ex:twofe}\small
	Recall
	\begin{align}
		\bs{A}
		=
		\begin{pmatrix}
			1 & 1 \\
			2 & 1 \\
			2 & 1 \\
			1 & 2 \\
			1 & 2
		\end{pmatrix},
		\notag
	\end{align}
	and hence $m - \textup{rank} ( ( \bs{A}, \bs{\iota}_5 ) ) =2$, so there is (up to scale) at most two admissible variance weights vectors. It was previously established that there is a unique MVUJ estimator associated with this design which is given by
	\begin{align}
		\tilde{\varphi}
		=
		3 \hat{\varphi}^{(0)}
		-
		\frac{1}{2} \hat{\varphi}^{(1)}
		-
		\frac{1}{2} \hat{\varphi}^{(2)}
		-
		\frac{1}{2} \hat{\varphi}^{(3)}
		-
		\frac{1}{2} \hat{\varphi}^{(4)}.
		\notag
	\end{align}
	Notice that $\bs{C} \bs{v}^\ast = \bs{\iota}_5$ and so the restriction $\bs{u}^\top \bs{C} \bs{v}^\ast = 0$ coincides with $\bs{u}^\top \bs{\iota}_5 = 0$. Thus, we may take two linearly independent variance weights vectors and then normalise them to satisfy \eqref{eq:uq-var-match}. For instance, define
	\begin{align}
		\bs{u}_1^\ast
		\coloneqq
		\left(
		0,\,
		-\frac{1}{2},\,
		\frac{1}{2},\,
		0,\,
		0
		\right)^\top
		\qquad \text{and} \qquad
		\bs{u}_2^\ast
		\coloneqq
		\left(
		0,\,
		0,\,
		0,\,
		-\frac{1}{2},\,
		\frac{1}{2}
		\right)^\top.
		\notag
	\end{align}
	Then writing $\bs{U}^\ast \coloneqq ( \bs{u}_1^\ast, \bs{u}_2^\ast )$,
	\begin{align}
		\bs{U}^{\ast \top}\bs{A} = \bs{0}_{2 \times 2},
		\qquad
		\bs{U}^{\ast \top}\bs{\iota}_5 = \bs{0}_2,
		\quad \text{and} \quad
		\bs{U}^{\ast \top}\bs{C}\bs{U}^\ast = \bs{v}^{\ast \top}\bs{C}\bs{v}^\ast \bs{I}_2.
		\notag
	\end{align}
	The corresponding jackknife standard error is
	\begin{align}
		\tilde{\sigma}_2^2
		=
		\frac{1}{2}\big( \bs{u}_1^{\ast \top} \bs{\hat{\varphi}} \big)^2
		+
		\frac{1}{2}\big( \bs{u}_2^{\ast \top} \bs{\hat{\varphi}} \big)^2,
		\notag
	\end{align}
	and the jackknife $t_2$ statistic is given by
	\begin{align}
		\mathcal{J}_2
		\coloneqq
		\frac{ \tilde{\varphi} - \varphi }{ \tilde{\sigma}_2 }
		\xrightarrow{d}
		t_2.
		\notag
	\end{align}
\end{exampleagainhr}

\subsection{Jackknife $t^2_q$ Statistic}
In this section we present a statistic which nests those presented in Section \ref{jkt} and \ref{jktq} as special cases. In particular, we show how it is possible to construct a jackknife counterpart to Hotelling's $t^2$-statistic \citep{hotelling1931generalization}. This allows us to extend our method of inference to encompass vector parameters and to accommodate joint hypothesis tests. We begin by expanding Assumptions \ref{AAD} and \ref{AJK}. 

\begin{AssumptionAD*}\label{AADs}
	There exists an estimator $\hat{\bs{\varphi}}^{(0)}$ of a $p \times 1$ parameter of interest $\bs{\varphi}$ such that 
	\begin{align}
		r_{NT}( \hat{\bs{\varphi}}^{(0)} - \bs{\varphi} ) 
		=
		\bs{z}_{NT}^{(0)}
		+
		( \bs{\iota}_{R}^\top \otimes \bs{I}_{p} ) \bs{\mu}_{NT} 
		+
		\bs{\smallO}_{p}( 1 ), \notag 
	\end{align}
	as $N,T \rightarrow \infty$ and $N/T \rightarrow \gamma^2 \in ( 0, \infty )$, where $\bs{\mu}_{NT} \coloneqq ( \bs{\mu}_{NT,1}^\top, \ldots, \bs{\mu}_{NT,R}^\top )^\top$ with $\bs{\mu}_{NT,1} , \ldots, \bs{\mu}_{NT,R}$, which are nonstochastic, but possibly $N$ and $T$ dependent, and 
	\begin{align}
		\bs{z}_{NT}^{(0)}
		\xrightarrow{d}
		\mathcal{N}( \bs{0}_{p}, \bs{\Sigma} ), \notag 
	\end{align}
	with $\bs{\Sigma} \succ \bs{0}_{p \times p}$.
\end{AssumptionAD*}

Assumption \ref{AADs} is a natural generalisation of Assumption \ref{AAD}. Examples of such results may be found in the references cited in the discussion that surrounds that assumption. 

\begin{AssumptionJK*}\label{AJKs}
	Let $m \geq R + 2$ be a fixed integer. There exist $m - 1$ estimators $\hat{\bs{\varphi}}^{(1)}, \ldots, \hat{\bs{\varphi}}^{(m-1)}$ obtained from subsamples of data such that 
	\begin{align}
		r_{NT}(  \hat{\bs{\varphi}} - ( \bs{\iota}_{m} \otimes \bs{\varphi} ) )
		=
		\bs{z}_{NT}
		+
		( \bs{A} \otimes \bs{I}_{p} ) \bs{\mu}_{NT} 
		+
		\bs{\smallO}_{p}( 1 ), \notag 
	\end{align}
	as $N,T \rightarrow \infty$ and $N/T \rightarrow \gamma^2 \in ( 0, \infty )$, where 
	\begin{align}
		\hat{\bs{\varphi}} 
		\coloneqq
		\begin{pmatrix}
			\hat{\bs{\varphi}}^{(0)} \\
			\hat{\bs{\varphi}}^{(1)} \\
			\vdots \\
			\hat{\bs{\varphi}}^{(m-1)}
		\end{pmatrix},
		\qquad 
		\bs{z}_{NT}
		\xrightarrow{d}
		\mathcal{N}( \bs{0}_{mp}, ( \bs{C} \otimes \bs{\Sigma} ) ), \notag 
	\end{align}
	$\bs{A}$ and $\bs{C}$ are known matrices of dimension $m \times R$ and $m \times m$, respectively, $\textup{rank}(\bs{A}) = R$, $\bs{\iota}_m \notin \textup{col}( \bs{A} )$, $\textup{null}( (\bs{A}, \bs{\iota}_m )^\top) \not\subseteq \textup{null}(\bs{C})$, and $\bs{C} \succeq \bs{0}_{m \times m}$.
\end{AssumptionJK*}

Now, consider a vector-valued function $\bs{g}( \bs{\varphi} ) \in \mathbb{R}^{k \times 1}$ which is continuously differentiable in a neighbourhood of $\bs{\varphi}$ with derivative $\bs{G}( \bs{\varphi} ) \in \mathbb{R}^{k \times p}$. Let
\begin{align}
	\bs{h}( \bs{\hat{\varphi}} )
	\coloneqq
	\begin{pmatrix}
		\bs{g}( \bs{\hat{\varphi}}^{(0)} ) \\
		\bs{g}( \bs{\hat{\varphi}}^{(1)} ) \\
		\vdots \\
		\bs{g}( \bs{\hat{\varphi}}^{(m-1)} )
	\end{pmatrix}, \notag 
\end{align}
be a $mk \times 1$ vector. Under Assumption \ref{AADs} and \ref{AJKs}, using $\bs{A}$ and $\bs{C}$ it is possible to construct an MVUJ weights vector $\bs{v}^{\ast}$, as well as variance weights vectors $\bs{u}_{1}^{\ast}, \ldots, \bs{u}_{q}^{\ast}$, in the same manner as described in Sections \ref{MVUJ} and \ref{jktq}, respectively. Indeed, we will presume that it is possible to construct $q \geq k$ such variance weights vectors. Define the $k \times k$ matrix
\begin{align}
	\bs{\hat{\Sigma}}_{q} 
	\coloneqq
	\frac{1}{q} \sum_{ l = 1 }^{q}
	\left(
	( \bs{u}_{l}^{\ast\top} \otimes \bs{I}_{k} ) \bs{h}( \hat{\bs{\varphi}} )
	\right)
	\left(
	( \bs{u}_{l}^{\ast\top} \otimes \bs{I}_{k} ) \bs{h}( \hat{\bs{\varphi}}  )
	\right)^\top, \notag  
\end{align}
which generalises the jackknife variance described in Theorems \ref{thm:jt} and \ref{thm:jtq}. Using this we form the following statistic
\begin{align}
	\mathcal{J}^2_q 
	\coloneqq
	\left(
	( \bs{v}^{\ast\top} \otimes \bs{I}_{k} ) \bs{h}( \hat{\bs{\varphi}}  )
	\right)^\top
	\bs{\hat{\Sigma}}_{q}^{-1}
	\left(
	( \bs{v}^{\ast\top} \otimes \bs{I}_{k} ) \bs{h}( \hat{\bs{\varphi}}  )
	\right). \label{jt2q}
\end{align}
We refer to \eqref{jt2q} as the \textit{jackknife $t^2_q$-statistic}. Theorem \ref{thm:jk2q} below characterises the asymptotic distribution of this statistic. 

\begin{theorem}\label{thm:jk2q}
	Assume $\mathcal{V}\ \cap \ \textup{null}( \bs{C} ) = \emptyset$. Let $\bs{v}^\ast$ solve \eqref{prog}, and let $\bs{u}_1^\ast, \ldots, \bs{u}_q^\ast$ be such that with $\bs{U}^\ast \coloneqq ( \bs{u}_1^\ast, \ldots, \bs{u}_q^\ast )$ conditions \eqref{eq:uq-bias} - \eqref{eq:uq-var-match} hold. Further, assume that $q \geq k$, $\bs{G}( \bs{\varphi} ) \bs{\Sigma} \bs{G}( \bs{\varphi} )^{\top} \succ \bs{0}_{k \times k}$, and $\| ( \bs{A} \otimes \bs{I}_{p} ) \bs{\mu}_{NT} \|_2 = \mathcal{O}( 1 )$. Under Assumptions \ref{AADs} and \ref{AJKs},
	\begin{align}
		\mathcal{J}^2_q  \xrightarrow{d} t^2_{q,k} \quad \text{if} \quad \bs{g}( \bs{\varphi} ) = \bs{0}_{k \times 1}, \notag
	\end{align}
		as $N, T \rightarrow \infty$ and $N / T \rightarrow \gamma^2 \in ( 0, \infty )$, where $ t^2_{q,k}$ denotes Hotelling's $t^2$-distribution with parameters $q$ and $k$. 
\end{theorem}
\begin{proof}
	See Appendix \ref{appA}. 
\end{proof}

Using the result in Theorem \ref{thm:jk2q}, one can readily conduct hypothesis tests, obtain associated $p$-values, and construct confidence ellipsoids for the parameter of interest. Compared to Theorem \ref{thm:jt} and \ref{thm:jtq}, Theorem \ref{thm:jk2q} imposes a stronger condition on the bias term since it assumes
$\| ( \bs{A} \otimes \bs{I}_{p} ) \bs{\mu}_{NT} \|_2 = \mathcal{O}( 1 )$. This assumption can be dispensed with when $\bs{g}( \cdot )$ is linear.  

\begin{remark}
	Suppose $k = 1$ and take ${g}( {\varphi} ) = \varphi - \varphi_{0}$, where $\varphi_{0}$ is a hypothesised value of $\varphi$. Then $\bs{h}( \hat{\bs{\varphi}} )$ is an $m \times 1$ vector with $j$-th entry $\hat{\varphi}^{(j)} - \varphi_{0}$. Moreover, $\bs{\hat{\Sigma}}_{q}$ is a scalar and equals to
	\begin{align}
		\bs{\hat{\Sigma}}_{q}
		=
		\frac{1}{q}\sum_{l=1}^{q}
		\left(
		\bs{u}_{l}^{\ast\top}\bs{h}( \hat{\bs{\varphi}} )
		\right)^{2},
		\notag
	\end{align}
	while the numerator in \eqref{jt2q} reduces to $( \bs{v}^{\ast \top}\bs{h}( \hat{\bs{\varphi}} ) )^{2}$. Hence,
	\begin{align}
		\mathcal{J}^{2}_{q}
		=
		\frac{
			\left(
			\bs{v}^{\ast\top}\bs{h}( \hat{\bs{\varphi}} )
			\right)^{2}
		}{
			\frac{1}{q}\sum_{l=1}^{q}
			\left(
			\bs{u}_{l}^{\ast\top}\bs{h}( \hat{\bs{\varphi}} )
			\right)^{2}
		}
		=
		\left(
		\frac{
			\bs{v}^{\ast\top}\bs{h}( \hat{\bs{\varphi}} )
		}{
			\sqrt{ \frac{1}{q}\sum_{l=1}^{q}
				\left(
				\bs{u}_{l}^{\ast\top}\bs{h}( \hat{\bs{\varphi}} )
				\right)^{2} }
		}
		\right)^{2}.
		\notag
	\end{align}
	The bracketed expression in the rightmost display above is exactly the jackknife $t_{q}$-statistic of Section \ref{jktq}.
\end{remark}

\begin{remark}
	A $( 1 - \alpha )$ confidence ellipsoid for $\bs{\varphi}$ is given by the set of $\bs{\varphi}_0 \in \mathbb{R}^{p}$ satisfying
	\begin{align}
			\left(
			( \bs{v}^{\ast \top} \otimes \bs{I}_{p} ) 
			( \hat{\bs{\varphi}} - ( \bs{\iota}_{m} \otimes \bs{\varphi}_0 )  )		
			\right)^{\top}
			\bs{\hat{\Sigma}}_{q}^{-1}
			\left(
			( \bs{v}^{\ast \top} \otimes \bs{I}_{p} ) 
			( \hat{\bs{\varphi}} - ( \bs{\iota}_{m} \otimes \bs{\varphi}_0 )  )		
			\right)
			\leq
			t^{2}_{q,p,1-\alpha},
			\notag
	\end{align}
	where $t^{2}_{q,p,1-\alpha}$ denotes the $( 1 - \alpha )$ quantile of $t^{2}_{q,p}$.
\end{remark}

\begin{remark}
	Using the result 
	\begin{align}
		\frac{ q - k + 1 }{ k q } t^2_{q,k} \sim F_{k,q-k+1}, \notag 
	\end{align}
	for $q \geq k$, one can equally obtain critical values and compute $p$-values using the $F_{k,q-k+1}$ distribution rather than working directly with the $t^2_{q,k}$ distribution.
\end{remark}

\begin{remark}
	As $q \rightarrow \infty$, $t^2_{q,k} \rightarrow \chi^2_{k}$, and therefore, similar to the result in Section \ref{jktq}, one can obtain shorter confidence ellipsoids by increasing $q$ through the use of additional subsamples. 
\end{remark}

\section{Further Matter}

\subsection{Multi-dimensional Fixed Effects}\label{KWM}
We have focused thus far on two-dimensional panels since these remain the most commonly encountered in practice. However, panels with three or more dimensions are increasingly available, for example exporter-importer-year data and employer-employee-year data. Our framework extends naturally to multi-dimensional panels and provides a systematic way to construct bias-corrected jackknife estimators and jackknife $t$-statistics.

Consider a multi-dimensional panel with dimension sizes $N_1,\ldots,N_K$ and let $N \coloneqq N_1 \times \ldots \times N_{K}$. We generalise Assumptions \ref{AAD} and \ref{AJK}. 

\begin{AssumptionADD}\label{ADD}
	There exists an estimator of $\varphi$, denoted by $\hat{\varphi}^{(0)}$, such that 
	\begin{align}
		r_{N} \big( \hat{\varphi}^{(0)} - \varphi \big)
		=
		z_{N}
		+
		\sum_{r=1}^R \mu_{r,N}
		+
		\smallO_{p}( 1 ), \label{forms}
	\end{align}
	as $N_{\min} \rightarrow \infty$ along some sequence $\{ ( N_1 , \ldots , N_K ) \}$, where $\mu_{1,N}, \ldots, \mu_{R,N}$ are nonstochastic, but possibly $N_1, \ldots N_{K}$ dependent, and 
	\begin{align}
		z_{N}
		\xrightarrow{d}
		\mathcal{N}( 0, \sigma^2 ). \notag 
	\end{align} 
\end{AssumptionADD}

Results of the form \eqref{forms} have been established for Poisson models with three-way fixed effects in \cite{weidner2021bias}, and for a broader class of nonlinear models with three-way fixed effects in \cite{czarnowske2025debiased}. Moreover, there are other three-way models, and $K$-way  models more generally, for which results of the form \ref{ADD} could reasonably be supposed to hold. 

\begin{AssumptionJKD}\label{JKD}
	Let $m \geq R + 2$ be a fixed integer. There exist $m - 1$ estimators $\hat{\varphi}^{(1)}, \ldots, \hat{\varphi}^{(m-1)}$ obtained from subsamples of data such that
	\begin{align}
		r_{N} \big( \bs{\hat{\varphi}} - \varphi \bs{\iota}_m \big)
		&=
		\bs{z}_{N}
		+
		\bs{A} \bs{\mu}_{N}
		+
		\bs{\smallO}_p( 1 ),
		\notag
	\end{align}
	as $N_{\min} \rightarrow \infty$ along some sequence $\{ ( N_1, \ldots, N_K ) \}$, where
	\begin{align}
		\bs{z}_{N}
		\xrightarrow{d}
		\mathcal{N} \big( \bs{0}_m, \sigma^2 \bs{C} \big),
		\notag
	\end{align}
	$\bs{\hat{\varphi}} \coloneqq ( \hat{\varphi}^{(0)}, \ldots, \hat{\varphi}^{(m-1)} )^\top$, $\bs{A}$ and $\bs{C}$ are known matrices of dimension $m \times R$ and $m \times m$, respectively, $\textup{rank}(\bs{A}) = R$, $\bs{\iota}_m \notin \textup{col}( \bs{A} )$, $\textup{null}( (\bs{A}, \bs{\iota}_m )^\top) \not\subseteq \textup{null}(\bs{C})$, $\bs{C} \succeq \bs{0}_{m \times m}$, and $\bs{\mu}_{N} \coloneqq (\mu_{1,N}, \ldots, \mu_{R,N})^\top$.
\end{AssumptionJKD}
The leading order bias of estimators of fixed effects models often scales in a predictable way with sample size, such that the matrix $\bs{A}$ can be determined from knowledge only of the order of the bias. In the context of a two-way fixed effects model, we usually expect the (leading) order bias of an estimator to be of the form
\begin{align}
	\sqrt{ \frac{ N_1 }{ N_2 } } b_1
	+
	\sqrt{ \frac{ N_2 }{ N_1 } } b_2.
	\notag
\end{align}
Heuristically, the rate can be understood as being determined by the relative order of the standard error and the dimension of the fixed effects. For example, with fixed effects $\bs{\lambda}^{(1)}$ and $\bs{\lambda}^{(2)}$, where $\textup{dim}( \bs{\lambda}^{(1)} ) = N_1$ and $\textup{dim}( \bs{\lambda}^{(2)} ) = N_2$, and with the order of the standard error as $\sqrt{N}$,
\begin{align}
	\sqrt{ \frac{ N_1 }{ N_2 } }
	+
	\sqrt{ \frac{ N_2 }{ N_1 } }
	=
	\frac{
	\textup{dim}\big( \bs{\lambda}^{(1)} \big)
	+
	\textup{dim}\big( \bs{\lambda}^{(2)} \big)
	}{
	\sqrt{N}
	}.
	\notag
\end{align}
As another example, consider $K = 3$. With fixed effects $\bs{\lambda}^{(1)}$, $\bs{\lambda}^{(2)}$, and $\bs{\lambda}^{(3)}$, where $\textup{dim}\big( \bs{\lambda}^{(1)} \big) = N_1$, $\textup{dim}\big( \bs{\lambda}^{(2)} \big) = N_2$, and $\textup{dim}\big( \bs{\lambda}^{(3)} \big) = N_3$, and with the order of the standard error as $\sqrt{ N }$,
\begin{align}
	\frac{
	\textup{dim}\big( \bs{\lambda}^{(1)} \big)
	+
	\textup{dim}\big( \bs{\lambda}^{(2)} \big)
	+
	\textup{dim}\big( \bs{\lambda}^{(3)} \big)
	}{
	\sqrt{N}
	}
	=&\
	\frac{ \textup{dim}\big( \bs{\lambda}^{(1)} \big) }{\sqrt{N}}
	+
	\frac{ \textup{dim}\big( \bs{\lambda}^{(2)} \big) }{\sqrt{N}}
	+
	\frac{ \textup{dim}\big( \bs{\lambda}^{(3)} \big) }{\sqrt{N}}
	\notag \\
	=&\
	\sqrt{ \frac{ N_1 }{ N_2 N_3 } }
	+
	\sqrt{ \frac{ N_2 }{ N_1 N_3 } }
	+
	\sqrt{ \frac{ N_3 }{ N_1 N_2 } }.
	\notag
\end{align}
Or indeed, the fixed effects may be present in multiple dimensions simultaneously, for example, $\bs{\lambda}^{(1,2)}$, $\bs{\lambda}^{(2,3)}$, and $\bs{\lambda}^{(1,3)}$, where $\textup{dim}( \bs{\lambda}^{(1,2)} ) = N_1 N_2$, $\textup{dim}( \bs{\lambda}^{(2,3)} ) = N_2 N_3$, and $\textup{dim}( \bs{\lambda}^{(1,3)} ) = N_1 N_3$, and with the order of the standard error as $\sqrt{N}$,
\begin{align}
	\frac{
	\textup{dim}\big( \bs{\lambda}^{(1,2)} \big)
	+
	\textup{dim}\big( \bs{\lambda}^{(2,3)} \big)
		+
	\textup{dim}\big( \bs{\lambda}^{(1,3)} \big)
	}{
	\sqrt{N}
	}
	=&\
	\frac{ \textup{dim}\big( \bs{\lambda}^{(1,2)} \big) }{\sqrt{N}}
	+
	\frac{ \textup{dim}\big( \bs{\lambda}^{(2,3)} \big) }{\sqrt{N}}
	+
	\frac{ \textup{dim}\big( \bs{\lambda}^{(1,3)} \big) }{\sqrt{N}}
	\notag \\
	=&\
	\sqrt{ \frac{ N_1 N_2 }{ N_3 } }
	+
	\sqrt{ \frac{ N_2 N_3 }{ N_1 } }
	+
	\sqrt{ \frac{ N_1 N_3 }{ N_2 } }.
	\notag
\end{align}
Applying this same heuristic we are able to determine the (likely) structure of the matrix $\bs{A}$ in multi-dimensional fixed effects models.
\begin{remark}	
	It can indeed be that the bias terms diverge under certain asymptotic sequences. Nonetheless, this is not ruled out by Assumption \ref{JKD} (nor indeed Assumption \ref{AJK}). 
\end{remark}

The following example illustrates the construction of an MVUJ estimator and jackknife $t$-statistic in a three-way fixed effects model. 

\begin{examplehr}[Three-way Effects]\label{ex:threefe-overlap}\small
	Consider the following three-way fixed effects model 
	\begin{align}
		y_{ i_1 i_2 i_3 }
		=
		\lambda_{ i_1 i_2 }
		+
		\gamma_{ i_2 i_3 }
		+
		\delta_{ i_3 i_1 }
		+
		\varepsilon_{ i_1 i_2 i_3 },
		\quad
		i_k = 1, \ldots, N_k, 
		\quad k = 1, 2, 3,
		\notag
	\end{align}
	where $\varepsilon_{ i_1 i_2 i_3 } \sim \textup{iid } \mathcal{N}( 0, \varphi )$ and
	$\varepsilon_{ i_1 i_2 i_3 }$, $\lambda_{ i_1 i_2 }$, $\gamma_{ i_2 i_3 }$, and $\delta_{ i_3 i_1 }$
	are mutually independent. Write $N \coloneqq N_1 N_2 N_3$ and let $\hat{\varphi}^{(0)}$ be the ML estimator of $\varphi$. As $N_{\min} \rightarrow \infty$ with $N_1 \propto N_2 \propto N_3$,
	\begin{align}
		\sqrt{N}\big( \hat{\varphi}^{(0)} - \varphi \big)
		=
		\frac{1}{\sqrt{N}}
		\sum_{i_1=1}^{N_1}
		\sum_{i_2=1}^{N_2}
		\sum_{i_3=1}^{N_3}
		\nu_{i_1 i_2 i_3}
		+
		\mu_{1,N} + \mu_{2,N} + \mu_{3,N}
		+
		\smallO_p(1),
		\notag
	\end{align}
	with 
	\begin{align}
		\mu_{1,N} \coloneqq - \varphi \frac{N_1N_2}{\sqrt{N}},
		\qquad
		\mu_{2,N} \coloneqq - \varphi \frac{N_2N_3}{\sqrt{N}},
		\qquad
		\mu_{3,N} \coloneqq  -\varphi \frac{N_3N_1}{\sqrt{N}}, \notag 
	\end{align}
	and 
	\begin{align}
		\frac{1}{\sqrt{N}}
		\sum_{i_1=1}^{N_1}
		\sum_{i_2=1}^{N_2}
		\sum_{i_3=1}^{N_3}
		\nu_{i_1 i_2 i_3}
		\xrightarrow{d}
		\mathcal{N}(0,2\varphi^2), \notag 
	\end{align}
	where $\nu_{i_1 i_2 i_3} \coloneqq \varepsilon_{i_1 i_2 i_3}^2 - \varphi$. Thus Assumption \ref{ADD} holds with $r_N = \sqrt{N}$, $R = 3$, and $\bs{\mu}_N \coloneqq( \mu_{1,N}, \mu_{2,N}, \mu_{3,N} )^\top$. Assume $N_1 = 2N_{1,2}$, $N_2 = 2N_{2,2}$, and $N_3 = 2N_{3,2}$ where $N_{1,2}$, $N_{2,2}$, and $N_{3,2}$ are integers. Define subsamples
\begin{align}
	\begin{aligned}
		S_ 1&= \{ 1, \ldots, N_{1,2} \} \times \{ 1, \ldots, N_2 \} \times \{ 1, \ldots, N_3 \},\\
		S_2 &= \{ N_{1,2} + 1, \ldots, N_1 \} \times\{ 1, \ldots, N_2 \} \times \{ 1, \ldots, N_3 \},\\
		S_3 &= \{ 1, \ldots, N_1 \} \times \{ 1, \ldots, N_{2,2}\} \times \{ 1, \ldots ,N_3 \},\\
		S_4 &= \{ 1, \ldots, N_1 \} \times \{ N_{2,2} + 1, \ldots, N_2 \} \times \{ 1, \ldots, N_3 \},\\
		S_5 &= \{ 1, \ldots, N_1 \} \times \{ 1, \ldots, N_2 \} \times \{ 1, \ldots, N_{3,2} \},\\
		S_6 &= \{ 1, \ldots, N_1 \} \times \{ 1, \ldots, N_2 \} \times \{ N_{3,2} + 1, \ldots, N_3 \}.
	\end{aligned}
	\notag
\end{align}
Let $\hat{\varphi}^{(j)}$ be the ML estimator computed on subsample $S_j$ for $j=1,\ldots,6$ and define
\begin{align}
	\bs{\hat{\varphi}}
	\coloneqq
	\big(
	\hat{\varphi}^{(0)}, \hat{\varphi}^{(1)}, \hat{\varphi}^{(2)},\hat{\varphi}^{(3)}, \hat{\varphi}^{(4)}, \hat{\varphi}^{(5)}, \hat{\varphi}^{(6)}
	\big)^\top.
	\notag
\end{align}
Then Assumption \ref{JKD} is satisfied with $m=7$,
\begin{align}
	\bs{A}
	=
	\begin{pmatrix}
		1 & 1 & 1\\
		1 & 2 & 1\\
		1 & 2 & 1\\
		1 & 1 & 2\\
		1 & 1 & 2\\
		2 & 1 & 1\\
		2 & 1 & 1
	\end{pmatrix}
	\quad
	\text{and}
	\quad
	\bs{C}
	=
	\begin{pmatrix}
		1 & 1 & 1 & 1 & 1 & 1 & 1\\
		1 & 2 & 0 & 1 & 1 & 1 & 1\\
		1 & 0 & 2 & 1 & 1 & 1 & 1\\
		1 & 1 & 1 & 2 & 0 & 1 & 1\\
		1 & 1 & 1 & 0 & 2 & 1 & 1\\
		1 & 1 & 1 & 1 & 1 & 2 & 0\\
		1 & 1 & 1 & 1 & 1 & 0 & 2
	\end{pmatrix}.
	\notag
\end{align}
In this design the MVUJ weights vector is unique and given by
\begin{align}
	\bs{v}^\ast
	=
	\left(
	4,\,
	-\frac{1}{2},\,
	-\frac{1}{2},\,
	-\frac{1}{2},\,
	-\frac{1}{2},\,
	-\frac{1}{2},\,
	-\frac{1}{2}
	\right)^\top,
	\notag
\end{align}
such that $\tilde{\varphi} \coloneqq \bs{v}^{\ast \top} \bs{\hat{\varphi}}$.
Moreover, $m - \textup{rank}( ( \bs{A},\bs{\iota}_7 ) ) = 3$, so it is possible to construct up to $q=3$ variance weights vectors. A convenient choice is
\begin{align}
	\bs{u}_1^\ast
	=
	\left(0,\,-\frac{1}{2},\,\frac{1}{2},\,0,\,0,\,0,\,0\right)^\top,
	\qquad
	\bs{u}_2^\ast
	=
	\left(0,\,0,\,0,\,-\frac{1}{2},\,\frac{1}{2},\,0,\,0\right)^\top,
	\notag
\end{align}
\vspace{-0.5cm}
\begin{align}
	\text{and} \quad \bs{u}_3^\ast
	=
	\left(0,\,0,\,0,\,0,\,0,\,-\frac{1}{2},\,\frac{1}{2}\right)^\top.
	\notag
\end{align}
Consequently, for $q\in\{1,2,3\}$ one may take $\bs{U}^\ast\coloneqq(\bs{u}_1^\ast,\ldots,\bs{u}_q^\ast)$ and form
\begin{align}
	\mathcal{J}_q
	\coloneqq
	\frac{\tilde{\varphi}-\varphi}{\tilde{\sigma}_q}
	\xrightarrow{d}
	t_q
	\quad
	\text{with}
	\quad
	\tilde{\sigma}_q^2
	\coloneqq
	\frac{1}{q}\sum_{l=1}^q\big(\bs{u}_l^{\ast\top}\bs{\hat{\varphi}}\big)^2.
	\notag
\end{align}
\normalsize
\end{examplehr}

\subsection{Higher-order Bias Corrections}\label{HBC} 
In this section we discuss how higher-order bias can be embedded in our framework, thereby allowing us to construct higher-order MVUJ estimators. To fix ideas, suppose an estimator $\hat{\varphi}^{(0)}$ of a parameter of interest in a two-way fixed effects model admits an $\ell$-th order expansion
\begin{align}
	\sqrt{NT} \big( \hat{\varphi}^{(0)} - \varphi \big)
	&=
	z_{NT}
	+
	\sqrt{ \frac{N}{T} } \sum_{l=1}^{\ell}\frac{ b_{l}^{(T)} }{ T^{ ( l - 1 ) / 2 } }
	+
	\sqrt{ \frac{T}{N} } \sum_{l=1}^{\ell}\frac{ b_{l}^{(N)} }{ N^{ ( l - 1 ) / 2 } }
	+
	\smallO_p( 1 ),
	\notag
\end{align}
as $N, T \to \infty$ and $N / T \rightarrow \gamma^2 \in ( 0, \infty )$, where $z_{NT} \xrightarrow{d} \mathcal{N}( 0, \sigma^2 )$. Assume $T = ( \ell + 1 ) T_{\ell+1} $ where $T_{\ell+1}$ is an integer. Partition the time-dimension into $\ell + 1$ disjoint blocks of equal length with
\begin{align}
	I_l
	\coloneqq
	\{ ( l - 1 ) T_{\ell+1} + 1, \ldots, l T_{\ell+1} \},
	\qquad
	l = 1, \ldots, \ell + 1.
	\notag
\end{align}
For $j = 1, \ldots, \ell$ define subsamples
\begin{align}
	\begin{aligned}
		S_j
		\coloneqq
		\{ 1, \ldots, N \}
		\times
		\Big(
		\bigcup_{l=1}^{j} I_l
		\Big),
		\qquad
		|S_j|
		&=
		\kappa_j N T,
		\qquad
		\kappa_j
		\coloneqq
		\frac{ j }{ \ell + 1 }.
	\end{aligned}
	\notag
\end{align}
Analogously, assume $N = ( \ell + 1 ) N_{\ell+1}$ where $N_{\ell+1}$ is an integer, and partition the cross-section into $\ell + 1$ disjoint blocks of equal length with
\begin{align}
	J_l
	\coloneqq
	\{ ( l - 1 ) N_{\ell+1} + 1, \ldots, l N_{\ell+1} \},
	\qquad
	l = 1, \ldots, \ell + 1.
	\notag
\end{align}
For $j = \ell + 1, \ldots, 2 \ell$ define the subsamples
\begin{align}
	\begin{aligned}
		S_j
		\coloneqq
		\Big(
		\bigcup_{l=1}^{j-\ell} J_l
		\Big)
		\times
		\{ 1, \ldots, T \},
		\qquad
		|S_j|
		&=
		\kappa_{j-\ell} N T,
		\qquad
		\kappa_{j-\ell}
		\coloneqq
		\frac{ j - \ell }{ \ell + 1 }.
	\end{aligned}
	\notag
\end{align}
Figure \ref{fig1} below illustrates the subsampling scheme with $\ell = 2$.
\vspace{0.5cm}
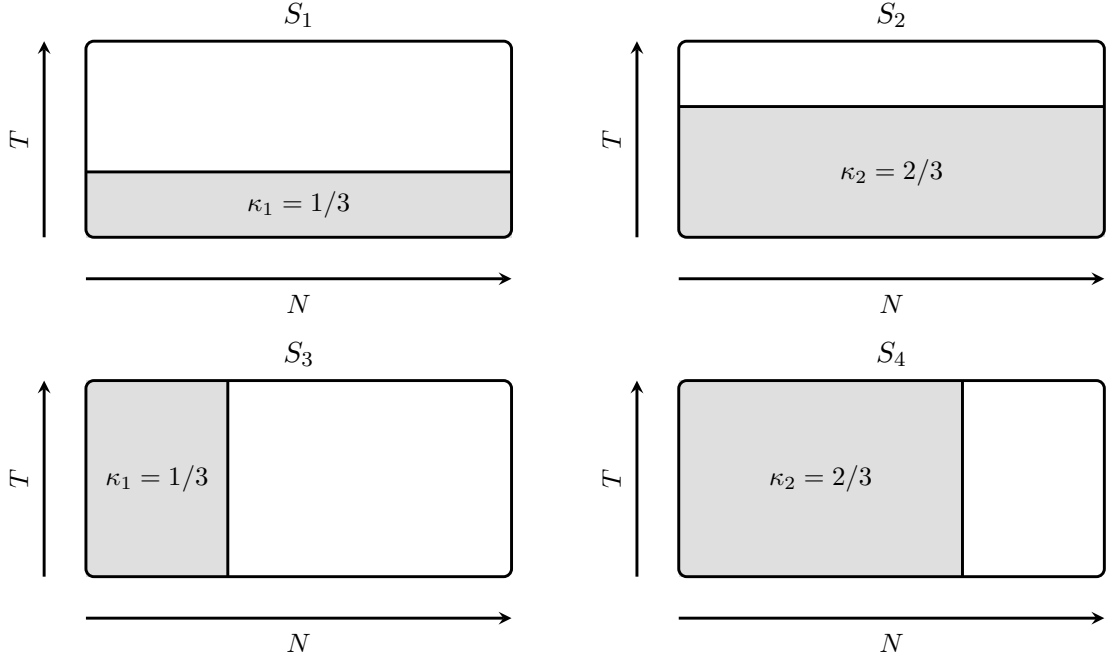
\begin{figure}[H]
	\centering
	\caption{Subsamples for $\ell = 2$.} \label{fig1}
	\begin{tikzpicture}[x=1cm,y=1cm, font=\small, >=stealth]
		\tikzset{
			panel/.style={draw, line width=1.1pt, rounded corners=3pt},
			split/.style={draw, line width=1.1pt},
			axis/.style={draw, ->, line width=1.1pt},
			inc/.style={fill=gray!25},
			txt/.style={font=\bfseries},
		}
		
		\pgfmathsetmacro{\W}{5.6}
		\pgfmathsetmacro{\H}{2.6}
		
		\pgfmathsetmacro{\kone}{1/3}
		\pgfmathsetmacro{\ktwo}{2/3}
		
		\pgfmathsetmacro{\tcutone}{\H*\kone}
		\pgfmathsetmacro{\tcuttwo}{\H*\ktwo}
		\pgfmathsetmacro{\ncutone}{\W*\kone}
		\pgfmathsetmacro{\ncuttwo}{\W*\ktwo}
		
		\pgfmathsetmacro{\rowgap}{4.5}
		
		\pgfmathsetmacro{\colgap}{7.8} 
		
		\begin{scope}[shift={(0,0)}]
			\path[inc] (0,0) rectangle (\W,\tcutone);
			\draw[panel] (0,0) rectangle (\W,\H);
			\draw[split] (0,\tcutone) -- (\W,\tcutone);
			
			\draw[axis] (0,-0.55) -- (\W,-0.55);
			\node at (\W/2,-0.88) {$N$};
			\draw[axis] (-0.55,0) -- (-0.55,\H);
			\node[rotate=90] at (-0.88,\H/2) {$T$};
			
			\node[txt] at (\W/2,\H+0.35) {$S_{ 1 }$};
			\node at (\W/2,0.5*\tcutone) {$\kappa_1 = 1 / 3$};
		\end{scope}
		
		\begin{scope}[shift={(\colgap,0)}]
			\path[inc] (0,0) rectangle (\W,\tcuttwo);
			\draw[panel] (0,0) rectangle (\W,\H);
			\draw[split] (0,\tcuttwo) -- (\W,\tcuttwo);
			
			\draw[axis] (0,-0.55) -- (\W,-0.55);
			\node at (\W/2,-0.88) {$N$};
			\draw[axis] (-0.55,0) -- (-0.55,\H);
			\node[rotate=90] at (-0.88,\H/2) {$T$};
			
			\node[txt] at (\W/2,\H+0.35) {$S_{ 2 }$};
			\node at (\W/2,0.5*\tcuttwo) {$\kappa_2 = 2 / 3$};
		\end{scope}
		
		\begin{scope}[shift={(0,-\rowgap)}]
			\path[inc] (0,0) rectangle (\ncutone,\H);
			\draw[panel] (0,0) rectangle (\W,\H);
			\draw[split] (\ncutone,0) -- (\ncutone,\H);
			
			\draw[axis] (0,-0.55) -- (\W,-0.55);
			\node at (\W/2,-0.88) {$N$};
			\draw[axis] (-0.55,0) -- (-0.55,\H);
			\node[rotate=90] at (-0.88,\H/2) {$T$};
			
			\node[txt] at (\W/2,\H+0.35) {$S_{3 }$};
			\node at (0.5*\ncutone,\H/2) {$\kappa_1 = 1 / 3$};
		\end{scope}
		
		\begin{scope}[shift={(\colgap,-\rowgap)}]
			\path[inc] (0,0) rectangle (\ncuttwo,\H);
			\draw[panel] (0,0) rectangle (\W,\H);
			\draw[split] (\ncuttwo,0) -- (\ncuttwo,\H);
			
			\draw[axis] (0,-0.55) -- (\W,-0.55);
			\node at (\W/2,-0.88) {$N$};
			\draw[axis] (-0.55,0) -- (-0.55,\H);
			\node[rotate=90] at (-0.88,\H/2) {$T$};
			
			\node[txt] at (\W/2,\H+0.35) {$S_{ 4}$};
			\node at (0.5*\ncuttwo,\H/2) {$\kappa_2 = 2 / 3$};
		\end{scope}
		
	\end{tikzpicture}%
\end{figure}
Let $\hat{\varphi}^{(j)}$ be the estimator computed on subsample $S_j$ for $j = 1, \ldots, m$ where $m \coloneqq 2\ell$. Suppose each subsample estimator admits the expansion
\begin{align}
	\sqrt{NT} \big( \hat{\varphi}^{(j)} - \varphi \big)
	&=
	z_{j,NT}
	+
	\sqrt{ \frac{N}{T} } \sum_{l=1}^{\ell} \kappa_j^{-l/2} \frac{ b_{l}^{(T)} }{ T^{ ( l - 1 ) / 2 } }
	+
	\sqrt{ \frac{T}{N} } \sum_{l=1}^{\ell} \frac{ b_{l}^{(N)} }{ N^{ ( l - 1 ) / 2 } }
	+
	\smallO_p( 1 ),
	\notag
\end{align}
for $j = 1, \ldots, \ell$, and
\begin{align}
	\sqrt{NT} \big( \hat{\varphi}^{(j)} - \varphi \big)
	&=
	z_{j,NT}
	+
	\sqrt{ \frac{N}{T} } \sum_{l=1}^{\ell} \frac{ b_{l}^{(T)} }{ T^{ ( l - 1 ) / 2 } }
	+
	\sqrt{ \frac{T}{N} } \sum_{l=1}^{\ell} \kappa_{j-\ell}^{-l/2} \frac{ b_{l}^{(N)} }{ N^{ ( l - 1 ) / 2 } }
	+
	\smallO_p( 1 ),
	\notag
\end{align}
for $j = \ell + 1, \ldots, m$. Now, define $\bs{\hat{\varphi}} \coloneqq ( \hat{\varphi}^{(0)}, \hat{\varphi}^{(1)}, \ldots, \hat{\varphi}^{(m)} )^\top$, $r_{NT}\coloneqq \sqrt{NT}$, $R \coloneqq 2\ell$, and
\begin{align}
	\bs{\mu}_{NT}
	\coloneqq
	\Big(
	\mu_{1,NT},\, \ldots,\, \mu_{\ell,NT},\,
	\mu_{\ell+1,NT},\, \ldots,\, \mu_{R,NT}
	\Big)^\top,
	\notag
\end{align}
with
\begin{align}
	\mu_{r,NT}
	&\coloneqq
	\sqrt{ \frac{N}{T} } \frac{ b_{r}^{(T)} }{ T^{ ( r - 1 ) / 2 } },
	\qquad
	r = 1, \ldots, \ell,
	\notag \\
	\mu_{\ell+r,NT}
	&\coloneqq
	\sqrt{ \frac{T}{N} } \frac{ b_{r}^{(N)} }{ N^{ ( r - 1 ) / 2 } },
	\qquad
	r = 1, \ldots, \ell.
	\notag
\end{align}
Then
\begin{align}
	r_{NT} \big( \bs{\hat{\varphi}} - \varphi \bs{\iota}_{m+1} \big)
	=
	\bs{z}_{NT}
	+
	\bs{A} \bs{\mu}_{NT}
	+
	\smallO_p( 1 ),
	\notag
\end{align}
with
\begin{align}
	\bs{z}_{NT}
	\coloneqq
	\begin{pmatrix}
		z_{0,NT} \\
		z_{1,NT} \\
		\vdots   \\
		z_{m,NT}
	\end{pmatrix},
	\qquad
	\bs{A}
	\coloneqq
	\begin{pmatrix}
		\bs{\iota}_{\ell}^\top & \bs{\iota}_{\ell}^\top \\
		\bs{A}^{(1)} & \bs{\iota}_{\ell} \bs{\iota}_{\ell}^\top \\
		\bs{\iota}_{\ell} \bs{\iota}_{\ell}^\top & \bs{A}^{(2)}
	\end{pmatrix},
	\notag
\end{align}
and
\begin{align}
	\bs{A}^{(1)}
	\coloneqq
	\begin{pmatrix}
		\kappa_1^{-1/2}\ \ \kappa_1^{-1}\ \ \cdots\ \ \kappa_1^{-\ell/2} \\
		\kappa_2^{-1/2}\ \ \kappa_2^{-1}\ \ \cdots\ \ \kappa_2^{-\ell/2} \\
		\vdots \\
		\kappa_\ell^{-1/2}\ \ \kappa_\ell^{-1}\ \ \cdots\ \ \kappa_\ell^{-\ell/2}
	\end{pmatrix},
	\quad
	\bs{A}^{(2)}
	\coloneqq
	\begin{pmatrix}
		\kappa_1^{-1/2}\ \ \kappa_1^{-1}\ \ \cdots\ \ \kappa_1^{-\ell/2}\\
		\kappa_2^{-1/2}\ \ \kappa_2^{-1}\ \ \cdots\ \ \kappa_2^{-\ell/2}\\
		\vdots \\
		\kappa_\ell^{-1/2}\ \ \kappa_\ell^{-1}\ \ \cdots\ \ \kappa_\ell^{-\ell/2}
	\end{pmatrix},
	\notag
\end{align}
where $\kappa_j = \frac{ j }{ \ell + 1 }$. Finally, if, for example,
\begin{align}
	z_{NT}
	=
	\frac{1}{\sqrt{NT}} \sum_{(i,t) \in S_0} \nu_{it},
	\notag
\end{align}
with $\nu_{it} \sim \textup{iid}( 0, \sigma^2 )$ and, for $j = 1, \ldots, \ell$,
\begin{align}
	z_{j,NT}
	=
	\frac{\kappa_j^{-1}}{\sqrt{NT}} \sum_{(i,t) \in S_j} \nu_{it},
	\qquad
	z_{\ell+j,NT}
	=
	\frac{\kappa_j^{-1}}{\sqrt{NT}} \sum_{(i,t) \in S_{\ell+j}} \nu_{it},
	\notag
\end{align}
then $\bs{z}_{NT} \xrightarrow{d} \mathcal{N}( \bs{0}, \sigma^2 \bs{C} )$ with
\begin{align}
	\bs{C}
	=
	\begin{pmatrix}
		1 & \bs{\iota}_\ell^\top & \bs{\iota}_\ell^\top\\
		\bs{\iota}_\ell & \bs{C}^{(1)} & \bs{\iota}_\ell \bs{\iota}_\ell^\top \\
		\bs{\iota}_\ell & \bs{\iota}_\ell \bs{\iota}_\ell^\top & \bs{C}^{(2)}
	\end{pmatrix},
	\quad
	C^{(1)}_{j_1,j_2}
	=
	\frac{1}{ \max\big( \kappa_{j_1}, \kappa_{j_2} \big) },
	\quad \text{and} \quad
	C^{(2)}_{j_1,j_2}
	=
	\frac{1}{ \max\big( \kappa_{j_1}, \kappa_{j_2} \big) }. \notag 
\end{align}
for $j_1, j_2 = 1, \ldots, \ell$. With $\bs{A}$ and $\bs{C}$ determined, one may construct an $\ell$-th order corrected MVUJ estimator and the corresponding $\ell$-th order jackknife $t$-statistic as previously.

\begin{remark}
	In multi-dimensional panel models estimators may frequently exhibit biases of different orders, some of which may diverge. For example, \cite{czarnowske2025debiased} derive an asymptotic expansion for M-estimators of three-way fixed effects models. For ``interacted'' fixed effects specifications (as in Example \ref{ex:threefe-overlap}) they obtain an expansion of the form
	\begin{align}
		\sqrt{N} ( \hat{\varphi}^{(0)} - \varphi )
		=&\
		z_{N}
		+
		\sqrt{ \frac{ N_1 N_2 }{ N_3 } } \mu_{1,N}
		+
		\sqrt{ \frac{ N_2 N_3 }{ N_1 } } \mu_{2,N} \notag \\
		&+\ 
		\sqrt{ \frac{ N_3 N_1 }{ N_2 } } \mu_{3,N}
		+
		\frac{ \sqrt{ N_1 N_2 } }{ N_3 } \mu_{4,N}
		+
		\smallO_p( 1 ), \notag
	\end{align}
	as $N_{\min} \rightarrow \infty$, $N_1 / N_3 \rightarrow \gamma^2_{13}$ and $N_2 / N_3 \rightarrow \gamma^2_{23}$, with
	\begin{align}
		z_{N}  \xrightarrow{d} \mathcal{N}( 0, \sigma^2 ).\notag
	\end{align}
	In particular, in the absence of strict exogeneity, the first three bias terms may diverge, while the fourth is $\mathcal{O}(1)$ under the stated asymptotic regime. Nonetheless, by finding a set of subsamples which satisfies Assumption \ref{AJKs}, we can proceed to construct an MVUJ estimator and a jackknife $t$-statistic. For instance, assume $N_1 = 2 N_{1,2}$, $N_2 = 2 N_{2,2}$, and $N_{3} = 2 N_{3,2} = 3 N_{3,3}$ where $N_{1,2},N_{2,2},N_{3,2}$, and $N_{3,3}$ are integers. Define subsamples
	\begin{align}
		S_{1} &= \{ 1, \ldots, N_{1,2} \} \times \{ 1, \ldots, N_2 \} \times \{ 1, \ldots, N_{3} \}, \notag \\
		S_{2} &= \{ N_{1,2} + 1, \ldots, N_{1} \} \times \{ 1, \ldots, N_2 \} \times \{ 1, \ldots, N_{3} \}, \notag \\
		S_{3} &= \{ 1, \ldots, N_{1} \} \times \{ 1, \ldots, N_{2,2} \} \times \{ 1, \ldots, N_{3} \}, \notag \\
		S_{4} &= \{ 1, \ldots, N_{1} \} \times \{ N_{2,2} + 1, \ldots, N_2 \} \times \{ 1, \ldots, N_{3} \}, \notag \\
		S_{5} &= \{ 1, \ldots, N_{1} \} \times \{ N_{2,2} + 1, \ldots, N_2 \} \times \{ 1, \ldots, {N_{3,2}} \}, \notag \\
		S_{6} &= \{ 1, \ldots, N_{1} \} \times \{ N_{2,2} + 1, \ldots, N_2 \} \times \{ N_{3,2} + 1, \ldots, N_3 \}, \notag \\
		S_{7} &= \{ 1, \ldots, N_{1} \} \times \{ N_{2,2} + 1, \ldots, N_2 \} \times \{ 1, \ldots, N_{3,3} \}, \notag \\
		S_{8} &= \{ 1, \ldots, N_{1} \} \times \{ N_{2,2} + 1, \ldots, N_2 \} \times \{ N_{3,3} + 1, \ldots, 2 N_{3,3} \}, \notag \\
		S_{9} &= \{ 1, \ldots, N_{1} \} \times \{ N_{2,2} + 1, \ldots, N_2 \} \times \{ 2 N_{3,3} + 1, \ldots, N_{3} \}. \notag
	\end{align}
	Then
	\begin{align}
		\bs{A}^\top
		=
		\begin{pmatrix}
			1 & 1 & 1 & 1 & 1 & 2 & 2 & 3 & 3 & 3 \\
			1 & 2 & 2 & 1 & 1 & 1 & 1 & 1 & 1 & 1 \\
			1 & 1 & 1 & 2 & 2 & 1 & 1 & 1 & 1 & 1 \\
			1 & 1 & 1 & 1 & 1 & 2^{\frac{3}{2}} & 2^{\frac{3}{2}} & 3^{\frac{3}{2}} & 3^{\frac{3}{2}} & 3^{\frac{3}{2}}
		\end{pmatrix}. \notag
	\end{align}
	Since $\textup{rank}(\bs{A}) = R$ and $\bs{\iota}_{m}\notin\textup{col}(\bs{A})$, an asymptotically unbiased jackknife estimator can readily be obtained, as discussed in Section \ref{MVUJ}. Upon obtaining the associated matrix $\bs{C}$, an MVUJ estimator and jackknife $t$-statistic can further be obtained.\footnote{In the interests of space we do not present such constructions here.} This illustrates the more general point that our framework can be used to obtain jackknife bias-corrected estimators and jackknife $t$-statistics in settings where the biases of the estimator are of markedly different orders, of which the higher-order corrections discussed in this section are a particular case.
\end{remark}

We now illustrate the construction of a higher-order MVUJ estimator and associated jackknife $t$-statistic. 

\begin{examplehr}[Higher-Order Bias Corrections]\label{ex:twofe-ho}\small
	Consider the following two-way fixed effects model
	\begin{align}
		y_{it} = \lambda_i + \gamma_t + \varepsilon_{it},
		\qquad i = 1, \ldots, N,\ \ t = 1, \ldots, T,
		\notag
	\end{align}
	where $\varepsilon_{it} \sim \textup{iid}\ \mathcal{N}(0,\varphi)$, and $\varepsilon_{it}$, $\lambda_i$ and $\gamma_t$ are mutually independent. Let $\hat{\varphi}^{(0)}$ be the ML estimator of $\varphi$. One can decompose
	\begin{align}
		\sqrt{NT} \big( \hat{\varphi}^{(0)} - \varphi \big)
		=
		\frac{1}{\sqrt{NT}} 
		\sum_{i=1}^{N}
		\sum_{t=1}^{T} \nu_{it} 
		+ \mu_{1,NT} + \mu_{2,NT} + \mu_{3,NT} + \eta_{NT},
		\notag
	\end{align}
	where $\nu_{it} = \varepsilon_{it}^2 - \varphi$,
	\begin{align}
		\mu_{1,NT} = - \varphi \sqrt{ \frac{N}{T} },
		\qquad
		\mu_{2,NT} = - \varphi \sqrt{ \frac{T}{N} },
		\quad
		\text{and}
		\quad
		\mu_{3,NT} = \varphi \frac{1}{\sqrt{NT}}.
		\notag
	\end{align}
	As $N, T \rightarrow \infty$ and $N / T \rightarrow \gamma^2 \in ( 0, \infty )$,
	\begin{align}
		z_{NT}
		\coloneqq
		\frac{1}{\sqrt{NT}}
		\sum_{i=1}^{N}
		\sum_{t=1}^{T}
		\nu_{it}
		\xrightarrow{d}
		\mathcal{N}( 0, 2\varphi^2 ),
		\notag
	\end{align}
	and the remainder term $\eta_{NT}$, which is mean zero by construction, is $\smallO_p( 1 )$. Thus Assumption \ref{AAD} holds with $r_{NT} = \sqrt{NT}$ and $R = 3$. Now, assume $N = 3N_3$ and $T = 3T_3$ where $N_3$ and $T_3$ are integers, and define
	\begin{align}
		I_l \coloneqq \{ ( l - 1 ) T_3 + 1, \ldots, l T_3 \}
		\quad \text{and} \quad
		J_l \coloneqq \{ ( l - 1 ) N_3 + 1, \ldots, l N_3 \},
		\qquad
		l = 1, 2, 3.
		\notag
	\end{align}
	Define subsamples
	\begin{align}
		\begin{aligned}
			S_1 &\coloneqq \{ 1, \ldots, N \} \times I_1, \\
			\qquad
			S_2 &\coloneqq \{ 1, \ldots, N \} \times ( I_1 \cup I_2 ),\\
			S_3 &\coloneqq J_1 \times \{ 1, \ldots, T \}, \\
			\qquad
			S_4 &\coloneqq J_1 \times I_1.
		\end{aligned}
		\notag
	\end{align}
	Let $\hat{\varphi}^{(j)}$ be the ML estimator computed on subsample $S_j$ for $j = 1, \ldots, 4$, and define
	\begin{align}
		\bs{\hat{\varphi}}
		\coloneqq
		\big(
		\hat{\varphi}^{(0)},
		\hat{\varphi}^{(1)},
		\hat{\varphi}^{(2)},
		\hat{\varphi}^{(3)},
		\hat{\varphi}^{(4)}
		\big)^\top.
		\notag
	\end{align}
	Then Assumption \ref{AJK} is satisfied with $m = 5$, $R = 3$,
	\begin{align}
		\bs{A}
		=
		\begin{pmatrix}
			1 & 1 & 1\\
			3 & 1 & 3\\
			\frac{3}{2} & 1 & \frac{3}{2}\\
			1 & 3 & 3\\
			3 & 3 & 9
		\end{pmatrix},
		\quad \text{and} \quad
		\bs{C}
		=
		\begin{pmatrix}
			1 & 1 & 1 & 1 & 1\\
			1 & 3 & \frac{3}{2} & 1 & 3\\
			1 & \frac{3}{2} & \frac{3}{2} & 1 & \frac{3}{2}\\
			1 & 1 & 1 & 3 & 3\\
			1 & 3 & \frac{3}{2} & 3 & 9
		\end{pmatrix}.
		\notag
	\end{align}
	Moreover, $\bs{C} \succ \bs{0}_{m \times m}$ so the MVUJ weights vector is unique. Solving \eqref{prog} yields
	\begin{align}
		\bs{v}^\ast
		=
		\left(
		\frac{9}{4},\,
		-\frac{3}{4},\,
		0,\,
		-\frac{3}{4},\,
		\frac{1}{4}
		\right)^\top.
		\notag
	\end{align}
	Thus the MVUJ estimator is $\tilde{\varphi} \coloneqq \bs{v}^{\ast\top}\bs{\hat{\varphi}}$. Next, $m - \textup{rank}( ( \bs{A}, \bs{\iota}_5 ) ) = 1$ so the design supports at most $q = 1$ variance weights vector. Let
	\begin{align}
		\bs{u}_0
		\coloneqq
		\left(
		-\frac{3}{4},\,
		-\frac{1}{4},\,
		1,\,
		0,\,
		0
		\right)^\top,
		\notag
	\end{align}
	with which we may construct
	\begin{align}
		\bs{u}^\ast
		\coloneqq
		\left(
		\frac{ \bs{v}^{\ast\top} \bs{C} \bs{v}^\ast }{ \bs{u}_0^\top \bs{C} \bs{u}_0 }
		\right)^{\frac{1}{2}}
		\bs{u}_0.
		\notag
	\end{align}
	Thus the jackknife $t$-statistic is
	\begin{align}
		\mathcal{J}
		\coloneqq
		\frac{ \tilde{\varphi} - \varphi }{ \tilde{\sigma} }
		\xrightarrow{d}
		t_1,
		\notag
	\end{align}
	with $\tilde{\sigma}
	\coloneqq
	\left|
	\bs{u}^{\ast\top}\bs{\hat{\varphi}}
	\right|$.
\end{examplehr}

\section{Numerical Examples}
In this section we showcase both the performance and the versatility of our method of inference through two numerical examples. The first considers inference on a slope coefficient in a linear regression model with predetermined regressors and individual fixed effects. The second considers inference on a slope coefficient in a panel probit model with interactive fixed effects.  

\subsection{Linear Regression with Individual Effects}
In our first example the data are generated according to
\begin{align}
	y_{it} &= {\varphi}\, x_{it} + \lambda_i + \varepsilon_{it}, \qquad i=1,\ldots,N,\ \ t=1,\ldots,T, \notag\\
	x_{i1} &= 0, \notag \\
	x_{it} &= 1\{y_{i,t-1}>0\},\ \qquad \hspace{0.2cm} t\ge 2, \notag
\end{align}
where $\varphi=0.5$, $\lambda_i \sim \textup{iid}\ \mathcal{N}(0,1)$, and $\varepsilon_{it} \sim \textup{iid}\ \mathcal{N}(0,1)$. We study the LS estimator of $\varphi$, denoted by $\hat{\varphi}^{(0)}$. Assume $T = 2T_2$, $N = 2N_2$, and $N = 5N_5$ where $T_2$, $N_2$ and $N_5$ are integers. Consider the following subsamples
\begin{align}
	S_1
	&\coloneqq
	\{ 1, \ldots, N \} \times \{ 1, \ldots, T_2 \},
	\notag \\
	S_2
	&\coloneqq
	\{ 1, \ldots, N \} \times \{ T_2 + 1, \ldots, T \},
	\notag \\
	S_3
	&\coloneqq
	\{ 1, \ldots, N_2 \} \times \{ 1, \ldots, T \},
	\notag \\
	S_4
	&\coloneqq
	\{ N_2 + 1, \ldots, N \} \times \{ 1, \ldots, T \},
	\notag \\
	S_{4+g}
	&\coloneqq
	\{ ( g - 1 ) N_5 + 1, \ldots, g N_5 \} \times \{ 1, \ldots, T \},
	\qquad g = 1, \ldots, 5.
	\notag
\end{align}
Let $\hat{{\varphi}}^{(j)}$ denote the LS estimator computed on subsample $S_j$. We compare three different constructions. The first, JK(a), utilises subsamples $S_1$ and $S_2$, alongside $S_0$. With $\bs{\hat{{\varphi}}}_{(a)} \coloneqq ( \hat{{\varphi}}^{(0)}, \hat{{\varphi}}^{(1)}, \hat{{\varphi}}^{(2)} )$, the MVUJ estimator is
\begin{align}
	\tilde{{\varphi}}_{(a)}
	\coloneqq
	2 \hat{{\varphi}}^{(0)} - \frac{1}{2} \hat{{\varphi}}^{(1)} - \frac{1}{2} \hat{{\varphi}}^{(2)}.
	\notag
\end{align}
There is one admissible variance weights vector 
\begin{align}
	\bs{u}^\ast
	= 
	\left(
	0 ,\, \frac{1}{2} ,\,  -\frac{1}{2} 
	\right)^\top, \notag 
\end{align}
yielding the jackknife $t$-statistic 
\begin{align}
	\mathcal{J}{(a)}
	\coloneqq
	\frac{ \tilde{{\varphi}}_{(a)} - {\varphi}}{ \tilde{\sigma}_{(a)} }
	\xrightarrow{d} 
	t_{1}
	\quad
	\text{with}
	\quad
	\tilde{\sigma}_{(a)} = | \bs{u}^{\ast\top} \bs{\hat{\varphi}}_{(a)} |. \notag 
\end{align}
The second, JK(b), uses additional cross-sectional partitions. In particular $\bs{\hat{{\varphi}}}_{(b)} \coloneqq (\hat{{\varphi}}^{(0)}, \hat{{\varphi}}^{(1)},\hat{{\varphi}}^{(2)},\hat{{\varphi}}^{(3)},\hat{{\varphi}}^{(4)})^\top$.
The MVUJ estimator is not unique, so we use a convenient construction
\begin{align}
	\tilde{{\varphi}}_{(b)}
	\coloneqq
	\frac{2}{3}\hat{{\varphi}}^{(0)}
	- \frac{1}{2}\hat{{\varphi}}^{(1)}
	- \frac{1}{2}\hat{{\varphi}}^{(2)}
	+ \frac{2}{3}\hat{{\varphi}}^{(3)}
	+ \frac{2}{3}\hat{{\varphi}}^{(4)}.
	\notag
\end{align}
Two admissible variance weights vectors are
\begin{align}
	\bs{u}_{1}^\ast 
	=
	\left( 0, \frac{1}{2}, -\frac{1}{2}, 0, 0 \right)^\top
	\quad
	\text{and}
	\quad
	\bs{u}_{2}^\ast
	=
	\left( 0, 0, 0, \frac{1}{2}, -\frac{1}{2} \right)^\top,
	\notag 
\end{align}
yielding
\begin{align}
	\mathcal{J}{(b)}
	\coloneqq
	\frac{ \tilde{{\varphi}}_{(b)} - {\varphi}}{ \tilde{\sigma}_{(b)} }
	\xrightarrow{d} t_{2}
	\quad
	\text{with}
	\quad
	\tilde{\sigma}_{(b)}
	\coloneqq 
	\sqrt{  \frac{1}{2} ( \bs{u}_{1}^{\ast \top} \bs{\hat{\varphi}}_{(b)} )^2 +  \frac{1}{2} (\bs{u}_{2}^{\ast \top} \bs{\hat{\varphi}}_{(b)} )^2  }
	\notag 
\end{align}
Finally, JK(c) uses additional cross-sectional partitions. In particular,
\begin{align}
	\bs{\hat{{\varphi}}}_{(c)}
	\coloneqq
	(\hat{{\varphi}}^{(0)} ,\, 
	\hat{{\varphi}}^{(1)} ,\,
	\hat{{\varphi}}^{(2)} ,\,
	\hat{{\varphi}}^{(5)} ,\,
	\hat{{\varphi}}^{(6)} ,\,
	\hat{{\varphi}}^{(7)} ,\,
	\hat{{\varphi}}^{(8)} ,\,
	\hat{{\varphi}}^{(9)})^\top, \notag 
\end{align}
where $\hat{{\varphi}}^{(4+g)}$ is the LS estimator on the $g$-th cross-sectional fifth. The MVUJ estimator is not unique so we again use a convenient construction
\begin{align}
	\tilde{{\varphi}}_{{(c)}}
	\coloneqq
	\hat{{\varphi}}^{(0)}
	-\frac{1}{2} \hat{{\varphi}}^{(1)}
	-\frac{1}{2} \hat{{\varphi}}^{(2)}
	+\frac{1}{5} \sum_{g=1}^5 \hat{{\varphi}}^{(4+g)}.
	\notag
\end{align}
Five admissible variance weights vectors are
\begin{align}
	\bs{u}^\ast_{1}
	&=
	\left( 
	0 ,\, \frac{1}{2} ,\, -\frac{1}{2} ,\, 0 ,\, 0 ,\, 0 ,\, 0 ,\, 0
	\right)^\top, \notag \\
	\bs{u}^\ast_{2}
	&=
	\left( 
	0 ,\, 0 ,\, 0 ,\, \frac{1}{\sqrt{10}} ,\, - \frac{1}{\sqrt{10}} ,\, 0 ,\, 0 ,\, 0
	\right)^\top, \notag \\
	\bs{u}^\ast_{3}
	&=
	\left( 
	0 ,\, 0 ,\, 0 ,\, \frac{1}{\sqrt{30}} ,\, \frac{1}{\sqrt{30}} ,\, {- \frac{2}{\sqrt{30}}} ,\, 0 ,\, 0
	\right)^\top, \notag \\
	\bs{u}^\ast_{4}
	&=
	\left( 
	0 ,\, 0 ,\, 0 ,\, \frac{1}{\sqrt{60}} ,\, \frac{1}{\sqrt{60}} ,\, \frac{1}{\sqrt{60}} ,\, {- \frac{3}{\sqrt{60}}},\, 0
	\right)^\top, \notag \\
	\bs{u}^\ast_{5}
	&=
	\left( 
	0 ,\, 0 ,\, 0 ,\, \frac{1}{10} ,\, \frac{1}{10} ,\, \frac{1}{10} ,\, \frac{1}{10} ,\, -\frac{4}{10}
	\right)^\top. \notag
\end{align}
yielding
\begin{align}
	\mathcal{J}{(c)}
	\coloneqq
	\frac{\tilde{{\varphi}}_{(c)}-{\varphi}}{\tilde{\sigma}_{(c)}}
	\xrightarrow{d} t_{5}
	\quad \text{with} \quad
	\tilde\sigma_{(c)}
	\coloneqq
	{\sqrt{ \frac{1}{5} \sum_{l=1}^{5} \left( \bs{u}_l^{\ast \top} \bs{\hat{\varphi}}_{(c)} \right)^2 }}. \notag
\end{align}

Table {\ref{Tab1}} below reports empirical coverage probabilities and average length of a 95\% two-sided confidence interval for ${\varphi}$ constructed by inverting the different jackknife $t$-statistics outlined above. We also report empirical bias and standard error of the different estimators. We compare our jackknife confidence intervals with two alternative constructions. The first utilises the bias correction provided by \cite{hahn2011bias} which is based on an analytical characterisation. The second uses the moving block bootstrap scheme described in \cite{higgins2025inference}.
\begin{landscape}
	\begin{table}[p]
		\centering
		\small
		\setlength{\tabcolsep}{4pt}
		\renewcommand{\arraystretch}{1.25}
		
		\caption{Linear Regression with Individual Effects}
		\label{Tab1}
		
		\newsavebox{\TabOneBox}
		\sbox{\TabOneBox}{%
			\begin{tabular}{@{}cc lcccccccccc@{}}
				\toprule\toprule
				\multicolumn{2}{c}{} &
				\multicolumn{1}{c}{} &
				\multicolumn{1}{c}{LS} &
				\multicolumn{3}{c}{JK} &
				\multicolumn{3}{c}{HK} &
				\multicolumn{3}{c}{HJ} \\
				\cmidrule(lr){4-4}\cmidrule(lr){5-7}\cmidrule(lr){8-10}\cmidrule(lr){11-13}
				$N$ & $T$ &
				&  & JK(a) & JK(b) & JK(c)
				& \multirow{2}{*}{$h = 1$} & \multirow{2}{*}{$h = 2$} & \multirow{2}{*}{$h = 5$}
				& \multirow{2}{*}{$\ell = 1$} & \multirow{2}{*}{$\ell = 2$} & \multirow{2}{*}{$\ell = 5$} \\
				\cmidrule(lr){5-5}\cmidrule(lr){6-6}\cmidrule(lr){7-7}
				& & & & $t_{ 1 }$ & $t_{ 2 }$ & $t_{ 5 }$ & & & & & & \\
				\midrule
				
				\multirow{4}{*}{100} & \multirow{4}{*}{10} & Bias      & -0.1701 & 0.0150 & 0.0147 & 0.0150 & -0.0445 & -0.0955 & -0.1737 & -0.1696 & -0.0496 & -0.0170 \\
				&    & Std. Err. & 0.0793  & 0.0956 & 0.0957 & 0.0958 & 0.0797  & 0.0842  & 0.0819  & 0.0794  & 0.0890  & 0.1016 \\
				\cmidrule(lr){3-13}
				&    & Coverage  & 0.4124  & 0.9538 & 0.9455 & 0.9286 & 0.9031  & 0.7469  & 0.3961  & 0.7094  & 0.9913  & 0.9103 \\
				&    & Length    & 0.3045  & 2.1164 & 0.7039 & 0.4162 & 0.3045  & 0.3045  & 0.3045  & 0.4075  & 0.5114  & 0.3443 \\
				
				\midrule
				
				\multirow{4}{*}{250} & \multirow{4}{*}{20} & Bias      & -0.0910 & 0.0034 & 0.0033 & 0.0032 & -0.0190 & -0.0235 & -0.0431 & -0.0908 & -0.0275 & -0.0142 \\
				&    & Std. Err. & 0.0365  & 0.0401 & 0.0401 & 0.0401 & 0.0364  & 0.0377  & 0.0377  & 0.0365  & 0.0391  & 0.0416 \\
				\cmidrule(lr){3-13}
				&    & Coverage  & 0.2862  & 0.9513 & 0.9442 & 0.9375 & 0.9158  & 0.8865  & 0.7590  & 0.4496  & 0.9902  & 0.9914 \\
				&    & Length    & 0.1403  & 0.8438 & 0.2962 & 0.1826 & 0.1403  & 0.1403  & 0.1403  & 0.1654  & 0.2110  & 0.2446 \\
				
				\midrule
				
				\multirow{4}{*}{1000} & \multirow{4}{*}{80} & Bias      & -0.0245 & 0.0002 & 0.0002 & 0.0002 & -0.0043 & -0.0015 & -0.0027 & -0.0245 & -0.0093 & -0.0055 \\
				&     & Std. Err. & 0.0091  & 0.0093 & 0.0093 & 0.0093 & 0.0091  & 0.0092  & 0.0092  & 0.0091  & 0.0092  & 0.0094 \\
				\cmidrule(lr){3-13}
				&     & Coverage  & 0.2322  & 0.9539 & 0.9512 & 0.9470 & 0.9205  & 0.9453  & 0.9371  & 0.2547  & 0.9029  & 0.9713 \\
				&     & Length    & 0.0355  & 0.1877 & 0.0696 & 0.0446 & 0.0355  & 0.0355  & 0.0355  & 0.0367  & 0.0398  & 0.0434 \\
				
				\toprule\toprule
			\end{tabular}%
		}
		
		\usebox{\TabOneBox}
		
		\vspace{0.4em}
		
		\parbox{\wd\TabOneBox}{%
			\footnotesize
			\raggedright
			\textit{Notes:} LS denotes the uncorrected LS estimator, JK(a) - JK(c) denote the jackknife estimators described in the main text, HK denotes the analytical bias correction of \cite{hahn2011bias} with bandwidth parameter $h$, and HJ denotes the moving block bootstrap of \cite{higgins2025inference} with block length $\ell$. Reported are empirical bias and empirical standard error of the estimators, and empirical coverage and length of a 95\% two-sided confidence interval.
		}
	\end{table}
\end{landscape}
Overall jackknife bias correction (all schemes) outperforms both analytical bias correction and bootstrap bias correction in terms of bias reduction, with the latter two being quite sensitive to the choice of tuning parameters (bandwidth, block size), especially in small samples. On the other hand, a small inflation of the standard error is observed for the jackknife bias-corrected estimator when compared to the uncorrected LS estimator. This, however, quickly diminishes as sample size increases. In terms of coverage, those confidence intervals obtained by inverting a jackknife $t$-statistic considerably outperform their analytical and bootstrap counterparts, both in their nominal coverage, and in their reliability, with bootstrap confidence intervals exhibiting particular sensitivity to block length. The observed length of each jackknife confidence interval is as expected, with the use of additional variance weights vectors producing shorter confidence intervals. This is illustrative of the trade-off between the most simple constructions, which produce wider confidence intervals, and those which are more complex, but yield shorter confidence intervals. Nonetheless, it should be emphasised that even the most complex construction, JK(c), which produces a length comparable to other methods, is still computationally trivial in comparison to the bootstrap, does not require the use of cumbersome analytic formulae, and does not utilise any tuning parameters.

\subsection{Probit Model with Interactive Effects}
In our second experiment we consider a panel probit model with interactive fixed effects and either exogenous or predetermined covariates. The data are generated according to
\begin{align}
	y_{it}
	&=
	\mathbf{1} \left\{ 
	\varphi x_{it}
	+
	\lambda_i \gamma_t
	+
	\varepsilon_{it} \geq 0 \right\},
	\qquad
	i = 1, \ldots, N,
	\quad
	t = 1, \ldots, T,
	\notag
\end{align}
where $\varphi = 1$, and $\lambda_i \sim \textup{iid}\ \mathcal{N}( 0, 1 )$, $\gamma_t \sim \textup{iid}\ \mathcal{N}\big( 0, 1 \big)$, and $\varepsilon_{it} \sim \textup{iid}\ \mathcal{N}( 0, 1 )$ are mutually independent. We consider two designs for $x_{it}$
\begin{align}
	x_{it}
	=
	\mathbf{1} \left\{ \eta_{it} > 0 \right\}
	\qquad \text{or} \qquad
	x_{it}
	=&\
	\mathbf{1} \left\{ y_{i,t-1} > 0 \right\}
	+
	\mathbf{1} \left\{ \eta_{it} > 0 \right\},
	\notag \\
	x_{i1}
	=&\
	\mathbf{1} \left\{ \eta_{i1} > 0 \right\},
	\notag
\end{align}
where $\eta_{it} \sim \textup{iid}\ \mathcal{N} ( 0 , 1 )$ and is independent of $\lambda_{i}$, $\gamma_{t}$, and $\varepsilon_{it}$. Our first jackknife estimator is based on the subsamples
\begin{align}
	S_{ 1 }
	&\coloneqq
	\left\{ 1 , \ldots , N \right\}
	\times
	\left\{ 1 , \ldots , {T_{ 2 }} \right\} ,
	\notag \\
	S_{ 2 }
	&\coloneqq
	\left\{ 1 , \ldots , N \right\}
	\times
	\left\{ {T_{ 2 }} + 1 , \ldots , T \right\} ,
	\notag \\
	S_{ 3 }
	&\coloneqq
	\left\{ 1 , \ldots , {N_{ 2 }} \right\}
	\times
	\left\{ 1 , \ldots , T \right\} ,
	\notag \\
	S_{ 4 }
	&\coloneqq
	\left\{ {N_{ 2 }} + 1 , \ldots , N \right\}
	\times
	\left\{ 1 , \ldots , T \right\} ,
	\notag
\end{align} 
where we assume $N = 2 {N_{2}}$ and $T = 2 {T_2}$ where ${N_2}$ and ${T_2}$ are integers. Let $\bs{ {\hat{\varphi}} }_{(a)} \coloneqq ( {\hat{\varphi}}^{(0)}, {\hat{\varphi}}^{(1)}, {\hat{\varphi}}^{(2)}, {\hat{\varphi}}^{(3)}, {\hat{\varphi}}^{(4)} )^{ \top }$. We construct an MVUJ estimator
\begin{align}
	\tilde{{\varphi}}_{(a)}
	\coloneqq
	3 {\hat{\varphi}}^{(0)}
	-
	\frac{ 1 }{ 2 }
	\left( {\hat{\varphi}}^{(1)} + {\hat{\varphi}}^{(2)} + {\hat{\varphi}}^{(3)} + {\hat{\varphi}}^{(4)} \right).
	\notag
\end{align}
Two admissible variance weights vectors are
\begin{align}
	\bs{u}_1^\ast
	=
	\left( 0 ,\, \frac{1}{2} ,\, -\frac{1}{2} ,\, 0 ,\, 0 \right)^\top
	\quad
	\text{and}
	\quad
	\bs{u}_2^\ast
	=
	\left( 0,\, 0 ,\, 0 ,\, \frac{1}{2} ,\, -\frac{1}{2} \right)^\top, \notag 
\end{align}
with which we construct two jackknife $t$-statistics
\begin{align}
	\mathcal{J}{(a)}
	=
	\frac{ \tilde{{\varphi}}_{(a)} - {\varphi} }{ \tilde{\sigma}_{(a)} }
	\xrightarrow{d}
	t_{1} 
	\quad
	\text{with}
	\quad
	\tilde{\sigma}_{(a)}
	=
	| \bs{u}_1^{\ast \top} \bs{\hat{\varphi}}_{(a)} |, \notag 
\end{align}
and 
\begin{align}
	\mathcal{J}{(b)}
	=
	\frac{ \tilde{{\varphi}}_{(a)} - {\varphi} }{ \tilde{\sigma}_{(b)} }
	\xrightarrow{d}
	t_{ 2 } 
	\quad
	\text{with}
	\quad
	\tilde{\sigma}_{(b)}
	=
	\sqrt{ \frac{1}{2} ( \bs{u}_1^{\ast \top} \bs{\hat{\varphi}}_{(a)} )^2 + \frac{1}{2} ( \bs{u}_2^{\ast \top} \bs{\hat{\varphi}}_{(a)} )^2 }, \notag 
\end{align}
Our second set of subsamples utilise $S_{ 1 }$, $S_{ 2 }$, and
\begin{align}
	S_{ 4 + g }
	&\coloneqq
	\left\{ \left( g - 1 \right) N_{5} + 1, \ldots, g N_{5} \right\}
	\times
	\left\{ 1, \ldots, T \right\} ,
	\qquad
	g = 1, \ldots, 5 .
	\notag
\end{align}
where we assume $N = 5 N_{5}$ where $N_5$ is an integer. Let
\begin{align}
	\bs{ {\hat{\varphi}} }_{(c)}
	\coloneqq
	\left( {\hat{\varphi}}^{(0)} ,\, {\hat{\varphi}}^{(1)},\, {\hat{\varphi}}^{(2)},\, {\hat{\varphi}}^{(5)},\, {\hat{\varphi}}^{(6)},\, {\hat{\varphi}}^{(7)},\, {\hat{\varphi}}^{(8)},\, {\hat{\varphi}}^{(9)} \right)^{\top},
	\notag
\end{align}
where ${\hat{\varphi}}^{(4+g)}$ is computed on the $g$-th cross-sectional fifth $S_{4+g}$. We construct an MVUJ estimator
\begin{align}
	\tilde{{\varphi}}_{(c)}
	\coloneqq
	\frac{ 9 }{ 4 } {\hat{\varphi}}^{(0)}
	-
	\frac{ 1 }{ 2 } {\hat{\varphi}}^{(1)}
	-
	\frac{ 1 }{ 2 } {\hat{\varphi}}^{(2)}
	-
	\frac{ 1 }{ 20 } \sum_{ g = 1 }^{ 5 } {\hat{\varphi}}^{ ( 4 + g ) } ,
	\notag
\end{align}
together with the jackknife standard error
\begin{align}
	\tilde{\sigma}_{(c)}
	\coloneqq
	\sqrt{ \frac{ 1 }{ 5 } \sum_{ l = 1 }^{ 5 } ( \bs{u}_l^{\ast \top}\bs{ {\hat{\varphi}} }_{(c)} )^2 } ,
	\notag
\end{align}
where
\begin{align}
	\bs{u}_{1}^{\ast}
	&\coloneqq
	\left( 
	0,\, \frac{ 1 }{ 2 },\, -\frac{ 1 }{ 2 },\, 0,\, 0,\, 0,\, 0,\, 0
	\right)^{\top},
	\notag \\
	\bs{u}_{2}^{\ast}
	&\coloneqq
	\left( 
	0,\, 0 ,\, 0,\, \frac{ 1 }{ \sqrt{ 10 } },\, -\frac{ 1 }{ \sqrt{ 10 } },\, 0,\, 0,\, 0
	\right)^{\top},
	\notag \\
	\bs{u}_{3}^{\ast}
	&\coloneqq
	\left( 
	0 ,\, 0 ,\, 0 ,\, \frac{ 1 }{ \sqrt{ 30 } } ,\, \frac{ 1 }{ \sqrt{ 30 } } ,\, -\frac{ 2 }{ \sqrt{ 30 } } ,\, 0 ,\, 0
	\right)^{\top},
	\notag \\
	\bs{u}_{4}^{\ast}
	&\coloneqq
	\left( 
	0,\, 0,\, 0,\, \frac{ 1 }{ \sqrt{ 60 } },\, \frac{ 1 }{ \sqrt{ 60 } },\, \frac{ 1 }{ \sqrt{ 60 } },\, -\frac{ 3 }{ \sqrt{ 60 } },\, 0
	\right)^{\top},
	\notag \\
	\bs{u}_{5}^{\ast}
	&\coloneqq
	\left(
	0,\, 0,\, 0,\, \frac{ 1 }{ 10 },\, \frac{ 1 }{ 10 },\, \frac{ 1 }{ 10 },\, \frac{ 1 }{ 10 },\, -\frac{ 4 }{ 10 }
	\right)^{\top}.
	\notag
\end{align}
This yields a jackknife $t$-statistic
\begin{align}
	\mathcal{J}{(c)}
	\coloneqq
	\frac{ \tilde{{\varphi}}_{(c)} - {\varphi} }{ \tilde{\sigma}_{(c)} }
	\xrightarrow{d}
	t_{5} .
	\notag
\end{align}

Table~{\ref{Tab2}} below reports empirical coverage probabilities and length of a 95\% two-sided confidence interval for ${\varphi}$ constructed by inverting $\mathcal{J}{(a)}$, $\mathcal{J}{(b)}$, and $\mathcal{J}{(c)}$. We also report empirical bias and standard error of the bias-corrected estimates.

\clearpage
\begin{landscape}
	
	\begin{table}[p]
		\centering
		\small
		\setlength{\tabcolsep}{4pt}
		\renewcommand{\arraystretch}{1.25}
		
		\caption{Probit Model with Interactive Effects}
		\label{Tab2}
		
		\newsavebox{\TabTwoBox}
		\sbox{\TabTwoBox}{%
			\begin{tabular}{@{}cc lcccc cccc@{}}
				\toprule\toprule
				\multicolumn{2}{c}{} & &
				\multicolumn{4}{c}{Exog.} &
				\multicolumn{4}{c}{Pred.} \\
				\cmidrule(lr){4-7}\cmidrule(lr){8-11}
				$N$ & $T$ & &
				MLE & JK(a) & JK(b) & JK(c) &
				MLE & JK(a) & JK(b) & JK(c) \\
				\cmidrule(lr){4-4}\cmidrule(lr){5-5}\cmidrule(lr){6-6}\cmidrule(lr){7-7}
				\cmidrule(lr){8-8}\cmidrule(lr){9-9}\cmidrule(lr){10-10}\cmidrule(lr){11-11}
				& & &
				& $t_{ 1 }$ & $t_{ 2 }$ & $t_{ 5 }$ &
				& $t_{ 1 }$ & $t_{ 2 }$ & $t_{ 5 }$ \\
				\midrule
				
				\multirow{4}{*}{100} & \multirow{4}{*}{10} & Bias
				& 0.1197 & -0.0326 & -0.0326 & -0.0344
				& 0.1002 & -0.0556 & -0.0556 & -0.0573 \\
				& & Std. Err.
				& 0.0973 & 0.1012 & 0.1012 & 0.0992
				& 0.1446 & 0.1884 & 0.1884 & 0.1680 \\
				\cmidrule(lr){3-11}
				& & Coverage
				& -- & 0.9483 & 0.9585 & 0.9639
				& -- & 0.9192 & 0.9304 & 0.9173 \\
				& & Length
				& -- & 1.9973 & 0.8602 & 0.5807
				& -- & 1.8011 & 0.8617 & 0.5922 \\
				\midrule
				
				\multirow{4}{*}{250} & \multirow{4}{*}{20} & Bias
				& 0.0554 & -0.0028 & -0.0028 & -0.0023
				& 0.0612 & -0.0023 & -0.0023 & -0.0026 \\
				& & Std. Err.
				& 0.0372 & 0.0370 & 0.0370 & 0.0369
				& 0.0328 & 0.0435 & 0.0435 & 0.0386 \\
				\cmidrule(lr){3-11}
				& & Coverage
				& -- & 0.9517 & 0.9568 & 0.9617
				& -- & 0.9470 & 0.9575 & 0.9573 \\
				& & Length
				& -- & 0.7664 & 0.3055 & 0.1949
				& -- & 0.6068 & 0.2606 & 0.1619 \\
				\midrule
				
				\multirow{4}{*}{1000} & \multirow{4}{*}{80} & Bias
				& 0.0139 & -0.0004 & -0.0004 & -0.0003
				& 0.0155 & -0.0003 & -0.0003 & -0.0003 \\
				& & Std. Err.
				& 0.0085 & 0.0084 & 0.0084 & 0.0084
				& 0.0063 & 0.0062 & 0.0062 & 0.0062 \\
				\cmidrule(lr){3-11}
				& & Coverage
				& -- & 0.9465 & 0.9495 & 0.9535
				& -- & 0.9468 & 0.9517 & 0.9492 \\
				& & Length
				& -- & 0.1717 & 0.0657 & 0.0421
				& -- & 0.1273 & 0.0490 & 0.0312 \\
				
				\bottomrule\bottomrule
			\end{tabular}%
		}
		
		\usebox{\TabTwoBox}
		
		\vspace{0.4em}
		
		\parbox{\wd\TabTwoBox}{%
			\footnotesize
			\raggedright
			\textit{Notes:} Exog. and Pred. denote the exogenous and predetermined regressor, respectively, MLE is the uncorrected ML estimator, and JK(a) - JK(c) are the jackknife estimators described in the main text. Reported are empirical bias and empirical standard error of the estimators, and empirical coverage and length of a 95\% two-sided confidence interval.
		}
	\end{table}
	
\end{landscape}

Both jackknife bias-corrected estimators achieve substantial bias reduction. When the covariate is exogenous there is very little inflation of the standard error when compared to the uncorrected ML estimator. When the covariate is predetermined, inflation of the standard error is observed. This occurs due to moderate dependence between partitions of the data in finite samples which diminishes as sample size (in particular the time-dimension) increases. Overall coverage is consistently close to its nominal value, particularly in the design where the covariate is exogenous. As expected, the use of additional variance weights vectors to produce $t_q$-statistics with $q > 1$ significantly reduces the length of confidence intervals.

\section{Conclusion}
This paper develops a straightforward and computationally inexpensive method of inference for fixed effects models. More broadly, this paper provides a general framework for the systematic development of jackknife bias-corrected estimators and jackknife $t$-statistics, that accommodates a wide range of models and estimation approaches, multi-dimensional panels, and higher-order bias corrections. Thus far, we have largely proceeded from the point at which a suitable collection of subsamples is obtained. Ongoing work develops general principles for constructing such subsamples, including optimality. Results in this regard will appear in subsequent revisions. 

Beyond fixed effects models, there has been increased interest in automatic inference procedures in other settings where estimators exhibit asymptotic bias. Much of this literature emphasises bootstrap or analytic bias correction (see, e.g., \cite{cavaliere2024bootstrap}). Our method provides an alternative to the bootstrap, one which both exhibits a substantially reduced computational burden and enjoys a high degree of model agnosticism. As such, our method may be usefully employed elsewhere. 

One promising line of enquiry is two-stage estimators, where the final stage estimator inherits bias from the first stage. Examples of this include models with many first-stage covariates \citep{cattaneo2019two} and factor-augmented regressions with estimated factors \citep{goncalves2014estimating}. In such instances, we may derive jackknife bias-corrected estimators and jackknife $t$-statistics through the generation of subsamples. Another promising line of enquiry is semiparametric and nonparametric estimators where asymptotic bias can arise as a consequence of smoothing \citep{calonico2014robust,calonico2018effect}. In such instances, we may manipulate bandwidth in order to produce a collection of estimators which exhibit different linear combinations of the same bias terms. Ongoing work explores the applications of our method in these settings.   

\newpage

\begin{appendices}

\section{Proof of Theorem \ref{thm:jk2q}}\label{appA}
Since $\bs{g}( \cdot )$ is continuously differentiable in a neighbourhood of $\bs{\varphi}$, for each $j = 0, \ldots, m - 1$,
\begin{align}
	\bs{g}\left( \hat{\bs{\varphi}}^{(j)} \right)
	=
	\bs{g}\left( \bs{\varphi} \right)
	+
	\bs{\bar{G}}^{(j)}
	\left( \hat{\bs{\varphi}}^{(j)} - \bs{\varphi} \right), \label{eqndd11}
\end{align}
by the mean-value theorem, where $\bs{\bar{G}}^{(j)}$ is a $k \times p$ matrix whose $s$-th row is equal to the $s$-th row of $\bs{G}( \bs{\bar{\varphi}}_{s}^{(j)})$, with $\bs{\bar{\varphi}}_{s}^{(j)} = \alpha_{j,s} \hat{\bs{\varphi}}^{(j)} + ( 1 - \alpha_{j,s} ) \bs{\varphi}$ for some $\alpha_{j,s} \in ( 0, 1 )$. Define
\begin{align}
	\bs{\bar{G}}
	\coloneqq
	\textup{blkdiag}
	\left(
	\bs{\bar{G}}^{(0)},
	\ldots,
	\bs{\bar{G}}^{(m-1)}
	\right).
	\notag
\end{align}
Stacking \eqref{eqndd11} over $j$ and multiplying by $r_{NT}$ gives
\begin{align}
	r_{NT} \bs{h}( \hat{\bs{\varphi}} )
	=
	r_{NT} \bs{h}( \bs{\varphi} )
	+
	\bs{\bar{G}}
	r_{NT}\left( \hat{\bs{\varphi}} - ( \bs{\iota}_{m} \otimes \bs{\varphi} ) \right).
	\notag
\end{align}
If $\bs{g}( \bs{\varphi} ) = \bs{0}_{k \times 1}$, we have $\bs{h}( \bs{\varphi} ) = \bs{0}_{mk}$, and, since $\hat{\bs{\varphi}}^{(j)} - \bs{\varphi} = \bs{\smallO}_{p}( 1 )$ for $j = 0, \ldots, m - 1$ under Assumption \ref{AJKs}, continuity of $\bs{G}( \cdot )$ implies
\begin{align}
	\bs{\bar{G}}
	=
	( \bs{I}_{m} \otimes \bs{G}( \bs{\varphi} ) )
	+
	\bs{\smallO}_{p}( 1 ).
	\notag
\end{align}
Hence,
\begin{align}
	r_{NT} \bs{h}( \hat{\bs{\varphi}} ) 
	=&\
	( \bs{I}_{m} \otimes \bs{G}( \bs{\varphi} ) )\,
	r_{NT}\left( \hat{\bs{\varphi}} - ( \bs{\iota}_{m} \otimes \bs{\varphi} ) \right)
	+
	\bs{\smallO}_{p}( 1 ).
	\notag
\end{align}
Substituting Assumption \ref{AJKs} yields
\begin{align}
	r_{NT} \bs{h}( \hat{\bs{\varphi}} ) 
	=&\
	( \bs{I}_{m} \otimes \bs{G}( \bs{\varphi} ) ) \bs{z}_{NT} 
	+
	( \bs{A} \otimes \bs{G}( \bs{\varphi} ) ) \bs{\mu}_{NT}
	+
	\bs{\smallO}_{p}( 1 ).
	\notag 
\end{align}
Now, premultiplying by $( \bs{v}^{\ast\top} \otimes \bs{I}_{k} )$ gives
\begin{align}
	r_{NT} ( \bs{v}^{\ast\top} \otimes \bs{I}_{k} ) \bs{h}( \hat{\bs{\varphi}} ) 
	=&\
	( \bs{v}^{\ast\top} \otimes \bs{G}( \bs{\varphi} ) ) \bs{z}_{NT} 
	+
	\bs{\smallO}_{p}( 1 ),
	\notag 
\end{align}
since $\bs{v}^{\ast \top} \bs{A} = \bs{0}$. Similarly, for $l = 1,\ldots,q$,
\begin{align}
	r_{NT} ( \bs{u}^{\ast\top}_{l} \otimes \bs{I}_{k} ) \bs{h}( \hat{\bs{\varphi}} ) 
	=
	( \bs{u}^{\ast\top}_{l} \otimes \bs{G}( \bs{\varphi} ) ) \bs{z}_{NT} 
	+
	\bs{\smallO}_{p}( 1 ).
	\notag 
\end{align}
Let $\bs{\xi} \sim \mathcal{N}( \bs{0}_{mp}, ( \bs{C} \otimes \bs{\Sigma} ) )$. Then for any $m \times 1$ vectors $\bs{b}, \bs{d}$,
\begin{align}
	\text{Cov}\left(
	( \bs{b}^\top \otimes \bs{G}( \bs{\varphi} ) ) \bs{\xi},
	( \bs{d}^\top \otimes \bs{G}( \bs{\varphi} ) ) \bs{\xi}
	\right)
	=&\
	( \bs{b}^\top \otimes \bs{G}( \bs{\varphi} ) )(  \bs{C} \otimes \bs{\Sigma} ) ( \bs{d} \otimes \bs{G}^\top( \bs{\varphi} ) ) 
	\notag \\
	=&\
	( \bs{b}^\top \bs{C} \bs{d} ) \bs{G}( \bs{\varphi} ) \bs{\Sigma} \bs{G}( \bs{\varphi} )^\top.
	\notag 
\end{align}
Applying this with $\bs{b} = \bs{v}^{\ast}$ and $\bs{d} = \bs{u}_{l}^{\ast}$, and using 
\begin{align}
	\bs{u}^{\ast\top}_{l} \bs{C} \bs{v}^{\ast} = 0,
	\quad
	\bs{u}^{\ast\top}_{l} \bs{C} \bs{u}^{\ast}_{l} = \bs{v}^{\ast\top} \bs{C} \bs{v}^{\ast},
	\quad
	\text{and}
	\quad
	\bs{u}^{\ast\top}_{l_1} \bs{C} \bs{u}^{\ast}_{l_2} = 0\ \text{for}\ l_1 \neq l_2,
	\notag 
\end{align}
it then follows that
\begin{align}
	\begin{pmatrix}
		r_{NT} ( \bs{v}^{\ast\top} \otimes \bs{I}_{k} ) \bs{h}( \hat{\bs{\varphi}} ) \\
		r_{NT} ( \bs{u}^{\ast\top}_{1} \otimes \bs{I}_{k} ) \bs{h}( \hat{\bs{\varphi}} )  \\
		\vdots \\
		r_{NT} ( \bs{u}^{\ast\top}_{q} \otimes \bs{I}_{k} ) \bs{h}( \hat{\bs{\varphi}} ) 
	\end{pmatrix}
	\xrightarrow{d}
	\mathcal{N}\left( 
	\bs{0}_{(q+1)k},
	\bs{I}_{q+1} \otimes \left( \bs{v}^{\ast\top} \bs{C} \bs{v}^{\ast} \bs{G}( \bs{\varphi} ) \bs{\Sigma} \bs{G}( \bs{\varphi} )^\top \right) 
	\right),
	\label{eq11}
\end{align}
as $N,T \rightarrow \infty$ and $N / T \rightarrow \gamma^2 \in ( 0, \infty )$, where $\bs{v}^{\ast\top} \bs{C} \bs{v}^{\ast} > 0$ since $\mathcal{V}\ \cap\ \textup{null}( \bs{C} ) = \emptyset$, and $\bs{G}( \bs{\varphi} ) \bs{\Sigma} \bs{G}( \bs{\varphi} )^\top \succ \bs{0}_{k \times k}$ by assumption. Define
\begin{align}
	\bs{\omega}_{NT}^{(0)}
	\coloneqq
	\left( \bs{v}^{\ast\top} \bs{C} \bs{v}^{\ast} \bs{G}( \bs{\varphi} ) \bs{\Sigma} \bs{G}( \bs{\varphi} )^\top \right)^{-\frac{1}{2}}
	r_{NT} ( \bs{v}^{\ast\top} \otimes \bs{I}_{k} ) \bs{h}( \hat{\bs{\varphi}} ),
	\notag
\end{align}
and for $l=1,\ldots,q$,
\begin{align}
	\bs{\omega}_{NT}^{(l)}
	\coloneqq
	\left( \bs{v}^{\ast\top} \bs{C} \bs{v}^{\ast} \bs{G}( \bs{\varphi} ) \bs{\Sigma} \bs{G}( \bs{\varphi} )^\top \right)^{-\frac{1}{2}}
	r_{NT} ( \bs{u}^{\ast\top}_{l} \otimes \bs{I}_{k} ) \bs{h}( \hat{\bs{\varphi}} ).
	\notag
\end{align}
Then \eqref{eq11} implies
\begin{align}
	\begin{pmatrix}
		\bs{\omega}_{NT}^{(0)} \\
		\bs{\omega}_{NT}^{(1)} \\
		\vdots \\
		\bs{\omega}_{NT}^{(q)}
	\end{pmatrix}
	\xrightarrow{d}
	\mathcal{N}\left( \bs{0}_{(q+1)k}, \bs{I}_{(q+1)k} \right), \label{vecdis}
\end{align}
from which it follows that $\bs{\omega}_{NT}^{(0)}$ and $\bs{\omega}_{NT}^{(1)},\ldots,\bs{\omega}_{NT}^{(q)}$ are asymptotically identically and independently distributed. Now, by the above, it follows from the continuous mapping theorem that
\begin{align}
	\frac{1}{q}\sum_{l=1}^{q}\bs{\omega}_{NT}^{(l)}( \bs{\omega}_{NT}^{(l)})^\top
	\xrightarrow{d}
	\frac{1}{q} 
	\bs{W},
	\quad
	\text{with}
	\quad
	\bs{W}
	\sim
	\bs{\mathcal{W}}_k ( \bs{I}_k, q ), \notag 	
\end{align}
where $\bs{\mathcal{W}}_k ( \bs{I}_k, q )$ denotes the Wishart distribution with parameters $\bs{I}_{k}$ and $q$. Appealing once more to the continuous mapping theorem
\begin{align}
	\mu_{\min}\left( \frac{1}{q}\sum_{l=1}^{q}\bs{\omega}_{NT}^{(l)}( \bs{\omega}_{NT}^{(l)})^\top \right)
	\xrightarrow{d}
	\mu_{\min}\left( \frac{1}{q} \bs{W} \right), \notag 
\end{align}
as $N,T \rightarrow \infty$, where $
\mathbb{P}( \text{det}( \bs{W} ) \neq 0 ) = 1$. Hence, we may write  
\begin{align}
	\mathcal{J}^{2}_{q}
	=
	( \bs{\omega}_{NT}^{(0)} )^{\top}
	\left(
	\frac{1}{q}\sum_{l=1}^{q}\bs{\omega}_{NT}^{(l)}( \bs{\omega}_{NT}^{(l)})^\top
	\right)^{-1}
	\bs{\omega}_{NT}^{(0)}. \label{omeg} 
\end{align}
Now let $\bs{Z}_{0}, \bs{Z}_{1}, \ldots, \bs{Z}_{q}$ be independent with $\bs{Z}_{l} \sim \mathcal{N}( \bs{0}_{k}, \bs{I}_{k} )$ for $l = 0, 1, \ldots, q$. Using \eqref{vecdis} and \eqref{omeg},
\begin{align}
	\mathcal{J}^{2}_{q}
	\xrightarrow{d}
	\bs{Z}_{0}^{\top}
	\left(
	\frac{1}{q}\sum_{l=1}^{q}\bs{Z}_{l}\bs{Z}_{l}^\top
	\right)^{-1}
	\bs{Z}_{0}
	\sim 
	t^2_{q,k}, \notag 
\end{align}
using Theorem 5.2.2 in \cite{anderson2003multivariate}. 
\qed

\end{appendices}

\end{document}